\documentclass[runningheads]{llncs}

\usepackage[comma,authoryear,round,longnamesfirst,sectionbib]{natbib}

\usepackage{graphicx}
\usepackage{wrapfig}
\usepackage{booktabs}
\usepackage{amsmath,amsfonts,amssymb}
\usepackage{stmaryrd}
\usepackage[ruled,linesnumbered, noend]{algorithm2e}
\usepackage{turnstile}
\usepackage{mathrsfs}
\usepackage{enumitem}
\usepackage{amsthm}
\usepackage[super]{nth}
\usepackage[toc,page]{appendix}
\usepackage{listings}
\usepackage{mathtools}
\usepackage{csquotes}
\usepackage{multicol}
\usepackage{fancyvrb}
\usepackage{makecell}
\usepackage{mdframed}
\usepackage[toc,page]{appendix}
\usepackage{mathdots}
\usepackage{empheq}
\usepackage{appendix}

\usepackage{array}
\newcolumntype{L}[1]{>{\raggedright\let\newline\\\arraybackslash\hspace{0pt}}m{#1}}
\newcolumntype{C}[1]{>{\centering\let\newline\\\arraybackslash\hspace{0pt}}m{#1}}
\newcolumntype{R}[1]{>{\raggedleft\let\newline\\\arraybackslash\hspace{0pt}}m{#1}}

\usepackage{tikz}
\usetikzlibrary{automata, positioning, calc, shapes, arrows, fit}

\newenvironment{claimproof}{\par\noindent\underline{Proof:}}{\leavevmode\unskip\penalty9999 \hbox{}\nobreak\hfill\quad\hbox{$\blacksquare$}}

\usepackage[unicode=true, bookmarks=false, breaklinks=true, pdfborder={0 0 1}, colorlinks=true, citecolor=blue]{hyperref}

\usepackage{pifont}
\definecolor{MyGreen}{rgb}{0, 0.7, 0}
\definecolor{MyRed}{rgb}{0.8, 0, 0}

\newcommand{\tuple}[1]{\left\langle #1 \right\rangle}

\newcommand{\agentSet}{\mathcal{N}}
\newcommand{\objSet}{\mathcal{O}}

\newcommand{\instBoth}{\ensuremath{\mathcal{I}}}
\newcommand{\instGood}{\ensuremath{\mathcal{I}^+}}
\newcommand{\instChore}{\ensuremath{\mathcal{I}^-}}

\newcommand{\GFoGood}{GEF1$^+$}
\newcommand{\GFoChore}{GEF1$^-$}

\newcommand{\squared}[1]{\fbox{#1}}

\title{Almost Group Envy-free Allocation \\of Indivisible Goods and Chores}

\author{Haris Aziz\inst{1} \and 
Simon Rey\inst{2}}

\institute{UNSW Sydney and Data61 Sydney, Australia  \\
\email{haziz@cse.unsw.edu.au} \and
Sorbonne Universit{\'e} and ENS Paris-Saclay, Paris, France \\
\email{srey@ens-paris-saclay.fr}}

\makeatletter
\newcommand{\@chapapp}{\relax}%
\makeatother

\begin{document}
	\maketitle
	
	\begin{abstract}
		We consider a multi-agent resource allocation setting in which an agent's utility may decrease or increase when an item is allocated. We take the group envy-freeness concept that is well-established in the literature and present stronger and relaxed versions that are especially suitable for the allocation of indivisible items. Of particular interest is a concept called group envy-freeness up to one item (GEF1).
		We then present a clear taxonomy of the fairness concepts. 
		We study which fairness concepts guarantee the existence of a fair allocation under which preference domain.
		For two natural classes of additive utilities, we design polynomial-time algorithms to compute a GEF1 allocation. 
		We also prove that checking whether a given allocation satisfies GEF1 is coNP-complete when there are either only goods, only chores or both. 
	\end{abstract}

\section{Introduction}

Fair division deals with the problem of assigning items to agents in the fairest way possible. Many fairness concepts have been proposed for this and envy-freeness (EF) is viewed as the gold standard one. It stipulates that for any pair of agents, no agent should prefer what the other got instead of her own share. However,
the concept does not provide any fairness guarantees when comparing groups of agents. Envy-freeness also has a well-known tension with efficiency goals~\citep{CKKK12}. 

 %guaranteeing fairness should not be done at the expense of efficiency.
%, see \cite{CKKK12} for a thorough analysis of this phenomena.

\citet{BTD92} introduced \emph{group envy-freeness (GEF)}. It generalizes envy-freeness for equal-sized groups of agents instead of considering only pairs of agents. One particularly desirable aspect of group envy-freeness is that it implies both envy-freeness and Pareto-optimality, the central concepts for fairness and efficiency respectively. 

Recently, \citet{CFS+19} generalized group-envy freeness for indivisible items and groups of different size by introducing \emph{group-fairness (GF)}. Since they considered indivisible items, guaranteeing GF allocations is impossible. They thus proposed two relaxations of group-fairness. These relaxations are similar in spirit to the well-studied weakening of envy-freeness called \emph{envy-freeness up to one good (EF1)} which requires no envy between any two agents as long as one agent gets rids of one of her goods. Allocations satisfying the relaxations of group-fairness can always be satisfied and can be computed in pseudo-polynomial time.

In their paper, \citet{CFS+19} assumed that the items allocated are ``goods'' for which agents have positive utility. Therefore the concepts and results do not apply to allocations scenarios in which tasks or chores are to be divided among agents. Although the definition of group-fairness and group-envy freeness can be extended seamlessly when there are either only chores or goods and chores, their relaxations do not. In this paper, we take inspiration from the GF and GEF concepts and define several variants and relaxations of them that are well-defined for the more general setting of goods and chores. Our approach is similar in spirit to the work of \citet{ACI+18} who presented general definitions for fairness concepts that apply as well to the case of goods and for chores. 

\paragraph{Contributions} Our first conceptual contribution is to formalize relaxations of GEF for the case of goods and chores. We give general definitions which apply seamlessly to non-additive preference or even ordinal preferences. The main relaxation is GEF1 which we show to be incomparable to the relaxations of group-fairness (GF). 
A stronger counterpart of GEF1, called s-GEF1, is introduced when groups of different size are allowed.
%The definition applies seamlessly to non-additive preference or even ordinal preferences, it therefore differs from the relaxations of group-fairness which are incomparable. 
%HARIS: please avoid run ons. 
%We moreover introduce group-proportionality, another relaxation of GEF that generalizes proportionality to groups. % We also define stronger counterpart of GEF1 denoted by adding the prefix s- called s-GEF1.
%All these concepts are meaningful in their own right. 
We clarify the logical relations between these concepts through a clear taxonomy depicted in Figure \ref{fig:FainessCreiteriaLinks}.

We present two key existence and algorithmic results. First, we design a polynomial-time algorithm that always computes a GEF1 allocation for the case of identical utilities. The proof relies on interesting connections with two-sided matching and it invokes Hall's marriage theorem. 
We then focus on a natural class of mixed utilities called ternary symmetric utilities and design an algorithm that returns an allocation that satisfies GEF1. This algorithm involves network flows and makes use of several transformations of the utility functions. The results also provide additional insights on the connection between Nash social welfare and leximin welfare. 

We then show that GEF1 allocations do not always exist for monotone utilities even with only goods. We also show that allocations satisfying the stronger concepts and group-proportionality are not guaranteed to exist. 
	
We prove that checking whether a given allocation satisfies GEF1 is coNP-complete for the cases of only goods, only chores and both.

We give sketches of the proofs for most of our results to present the key ideas. The complete proofs are in the appendix.

%\paragraph{Related work} 

\section{Related Work}
Fair division is a dynamic field in both economics and computer science \citep{BrTa96a, Moul04, BCM16, LaRo16}. The prominent fairness concept is envy-freeness (EF) \citep{Fole67} but an EF allocation may not exist for indivisible items. Checking whether there exists an EF allocation is NP-complete even for additive binary utilities~\citep{AGMW15a}.  Several relaxations have been considered to overcome this, in particular envy-freeness up to one good (EF1) \citep{Budi11} which always exists \citep{LMM+04}. \citet{CKM+16} proved the existence of allocation satisfying both EF1 and Pareto-optimality, and introduced envy-freeness up to any good (EFX).

\citet{BTD92} generalized envy-freeness to groups of agents by introducing \emph{group envy-freeness} when items are divisible, extending the idea of \emph{coalition fairness} \citep{ScVi72}.
% and \emph{strict envy-freeness} \citep{Zhou92}. 
An allocation is group envy-free if it is not possible to reallocation the items from one group of agents to another in a way that would Pareto-dominates the current allocation. They proved existence of GEF allocations under some monotonicity assumptions and showed the equivalence between EF and GEF. \citet{Huss11} extended it to \emph{weak group envy-freeness} by considering weak Pareto-improvement. 
%In a similar spirit, but of distinct the issue of allocating divisible items to pre-existent groups consisting of heterogeneous agents \citep{MaSu17, SeSh16}.

Similar generalizations have been proposed with indivisible items. \citet{TLH+11} introduced \emph{envy-freeness of a group toward a group} when monetary transfers between agents are allowed. Later \citet{AlWa18} presented another definition of \emph{group envy-freeness} between groups of potentially different sizes. However, it relies on interpersonal comparisons which has received criticism in the social choice literature.
\citet{CFS+19} defined \emph{group-fairness} for indivisible goods, a definition similar to that of \citet{BTD92} but considering indivisible items and groups of different size. They also introduced two ``up to one'' relaxations of group-fairness for which they proved existence by using some variant of the Nash social welfare. Another line of work in the same spirit is to consider pre-existing groups of agents, taken as inputs of the procedures \citep{SeSu18, KSV19}. Similarly, \cite{BCEZ19} investigated the problem of allocating indivisible goods to agents partitioned into types.

The chore division problem \citep{Gard78}, extends the classical fair division setting for items that are considered as chores for some agents. \citet{BrTa96a} and  \citet{Sega18} investigated the cake-cutting problem in this setting. \citet{BMSY17} studied mixture of divisible goods and chores. Indivisible chores have also been considered \citep{ARSW17, BaKr17}. \citet{CKKK12} analysed the \emph{price of fairness} and showed several differences between goods and chores settings. 
%\citet{ARSW17} and \cite{BaKr17} considered Max-min share with chores, showing non-existence and providing approximation algorithms. 
\citet{ACI+18} presented a general framework for indivisible goods and chores. In particular, they provided a general definition for EF1 and algorithms for it.
%EF1 allocation for preferences defined over bundles of items. 

\section{Preliminaries}

Let $\agentSet$ be a set of $n$ agents and $\objSet$ a set of $m$ items. Agent $i \in \agentSet$ has preferences over sets of items, called bundle, represented by a utility function $u_i : 2^\objSet \rightarrow \mathbb{R}$. We emphasize that agents can evaluate a bundle positively or negatively. Preferences are said to be \emph{additive} if for every subset of items $O \subseteq \objSet$, we have $u_i(O) = \sum_{o \in O} u_i(o)$. We assume additive preferences throughout the paper except explicitly stated otherwise. Our definitions can be applied to non-additive preferences. An item $o \in \objSet$ is a good for $i$ if $u_i(o) \geq 0$ and chore for $i$ if $u_i(o) \leq 0$. 

Let $O \subseteq \objSet$ be a subset of items and $N \subseteq \agentSet$ a subset of agents. An a\emph{llocation} $\pi = \tuple{\pi_1, \ldots, \pi_{|N|}}$ over $O$ and $N$ is a vector of bundles $\pi_i \subseteq O$ for $i \in N$. It satisfies indivisibility of the items, $\forall i, j \in N \text{ s.t. } i \neq j, \pi_i \cap \pi_j = \emptyset$, and non-wastefulness: $\bigcup_{i \in N} \pi_i = O$. For a subset of agents $N \subseteq \agentSet$, we denote by $\pi_N = \bigcup_{i \in N} \pi_i$ the set of items held by agents in $N$. We write respectively $\pi_i^+$ and $\pi_i^-$ the sets of goods and chores in $\pi_i$.

For a subset of items $O \subseteq \objSet$ and a subset of agents $N \subseteq \agentSet$, we denote by $\Pi(O, N)$ the set of all the allocations over $O$ and $N$. If $O \neq \objSet$, an allocation $\pi \in \Pi(O, \agentSet)$ is called \emph{partial}. A triplet $I = \tuple{\agentSet, \objSet, (u_i)_{i \in \agentSet}}$ is an instance. $\instBoth$ is the set of all the instances, $\instGood$ the set of instances with only goods and $\instChore$ the set of instances with only chores.

\medskip

Let $\pi$ be an allocation, we say that an allocation $\pi'$ \emph{Pareto-dominates} $\pi$ if all agents are better off in $\pi'$ and at least one agent is strictly better off: $\forall i \in \agentSet, u_i(\pi_i') \geq u_i(\pi_i)$ and $\exists i \in \agentSet, u_i(\pi_i') > u_i(\pi_i)$. An allocation $\pi$ is said to be \emph{Pareto-optimal} if no other allocation Pareto-dominates it. Another common efficiency criterion is to maximize the Nash social welfare, defined as $\prod_{i \in \agentSet} |u_i(\pi_i)|$.

\medskip

In the following, we introduce the definitions of envy-freeness and its relaxations when dealing with goods and chores as presented by \citet{ACI+18}.

\begin{definition}[Envy-freeness (EF)]
	Let $I = \tuple{\agentSet, \objSet, (u_i)_{i\in\agentSet}} \in \instBoth$ be an instance with both goods and chores.
	An allocation $\pi \in \Pi(\objSet, \agentSet)$ is envy-free if and only if: $\forall i, j \in \agentSet, u_i(\pi_i) \geq u_i(\pi_j)$.
\end{definition}

It is well known that for some instances envy-free allocations does not exist. Take for example two agents and one item. Two different relaxation of envy-freeness can then be considered.

We say that an allocation $\pi \in \Pi(\objSet, \agentSet)$ is envy-free up to one item (EF1) if:
$$\forall i, j \in \agentSet, \exists O \subseteq \pi_i \cup \pi_j, |O| \leq 1, \text{ s.t. } u_i(\pi_i\backslash O) \geq u_i(\pi_j\backslash O).$$

Moreover, an allocation $\pi \in \Pi(\objSet, \agentSet)$ is envy-free up to any item (EFX) if:
$$\forall i, j \in \agentSet, \left\{\begin{array}{l}
\forall o \in \pi_i \text{ s.t. } u_i(\pi_i) - u_i(\pi_i \backslash \{o\}) < 0, \quad u_i(\pi_i \backslash\{o\}) \geq u_i(\pi_j)\\
\forall o \in \pi_j \text{ s.t. } u_i(\pi_j) - u_i(\pi_j \backslash \{o\}) > 0, \quad u_i(\pi_i) \geq u_i(\pi_j\backslash\{o\}).
\end{array}\right.$$

\medskip

Finally, an allocation $\pi$ is proportional (PROP) if every agent is allocated her proportional share: $\forall i \in \agentSet, u_i(\pi_i) \geq \frac{u_i(\objSet)}{n}$.

%\begin{definition}[Proportionality up to one item (PROP1)]
%	Let $I = \tuple{\agentSet, \objSet, (u_i)_{i\in\agentSet}} \in \instBoth$ be an instance with both goods and chores.
%	An allocation $\pi \in \Pi(\objSet, \agentSet)$ is proportional up to one item if and only if:
%	$$\forall i \in \agentSet, \exists o \in \objSet \text{ s.t. }
%	\left\{\begin{array}{l}
%	u_i(\pi_i) \geq \frac{u_i(\objSet \backslash \{o\})}{n}, \quad or, \\
%	u_i(\pi_i \backslash \{o\}) \geq \frac{u_i(\objSet)}{n}.
%	\end{array}\right.$$
%\end{definition}

\section{Fairness criteria for groups with goods and chores}

In this section, we present our first contributions: a general definition for group envy-freeness and its relaxations in the presence of goods and chores.

\begin{definition}[Group envy-freeness]
	Let $I = \tuple{\agentSet, \objSet, (u_i)_{i\in\agentSet}} \in \instBoth$ be an instance with goods and chores. An allocation $\pi \in \Pi(\objSet, \agentSet)$ is GEF if for every $S, T \subseteq \agentSet$ such that $|S| = |T| \neq 0$, there is no $\pi' \in \Pi\left(\pi_T, S\right)$, such that:
	$$\forall i \in S, \frac{|S|}{|T|}u_i\left(\pi'_i\right) \geq u_i(\pi_i),$$
	with one inequality being strict. 
	
	We call the concept s-GEF if we do not impose the condition $|S| = |T|$.
\end{definition}

In words, GEF states that there is no reallocation of $\pi_T$ to the agents in $S$ that would Pareto-dominates the current allocation for agents in $S$.

Note that s-GEF is equivalent to group-fairness \citep{CFS+19}. The name \emph{group envy-freeness} is taken from \citet{BTD92} who introduced it for divisible items.

In the same spirit of EF1 and EFX, we introduce ``up to one'' and ``up to any'' relaxations for group envy-freeness.

\begin{definition}[Group envy-freeness up to one item, GEF1]
	An allocation $\pi \in \Pi(\objSet, \agentSet)$ is GEF1 if for every $S, T \subseteq \agentSet$ where $|S|=|T|\neq 0$, for every $\pi' \in \Pi\left(\pi_T, S\right)$, and for every $i \in S$, there exists $O_i \subseteq \pi_i^- \cup \pi_i'^+$, $|O_i| \leq 1$, such that $\tuple{\frac{|S|}{|T|} u_i(\pi_i'\backslash O_i)}_{i \in S}$ does not Pareto-dominate $\tuple{u_i(\pi_i\backslash O_i)}_{i \in S}$.
	
	We talk about s-GEF1 if we do not impose the condition $|S| = |T|$.
\end{definition}

\begin{definition}[Group envy-freeness up to any item, GEFX]
	An allocation $\pi \in \Pi(\objSet, \agentSet)$ is GEFX if for every $S, T \subseteq \agentSet$ where $|S|=|T| \neq 0$, for every $\pi' \in \Pi\left(\pi_T, S\right)$, for every $i \in S$, for every $o_i \subseteq \pi_i^- \cup \pi_i'^+$, $\tuple{\frac{|S|}{|T|} u_i(\pi_i'\backslash \{o_i\})}_{i \in S}$ does not Pareto-dominate $\tuple{u_i(\pi_i\backslash \{o_i\})}_{i \in S}$.
	
	We call the concept s-GEFX if we do not impose the condition $|S| = |T|$.
\end{definition}

Observe that with additive utility functions and for instances in $\instGood$, GEF1 is equivalent to \emph{group fairness up to one good after} (GF1A) as defined by \citet{CFS+19}. However, they also proposed \emph{group fairness up to one good before} (GF1B) which is no longer relevant when there are chores since removing items cannot be done ``before''.
GEF1 can be seen as an argument in favour of GF1A.

Nevertheless, s-GEF1 is not equivalent to GF1A even when considering only goods because of the way allocations are compared. Formally, $\tuple{u_i(\pi_i \cup \{o\})}_{i \in S}$ is compared to $\tuple{\frac{|S|}{|T|} u_i(\pi_i')}_{i \in S}$ in GF1A while s-GEF1 compares $\tuple{u_i(\pi_i)}_{i \in S}$ and $\tuple{\frac{|S|}{|T|} u_i(\pi_i'\backslash\{o\})}_{i \in S}$. The factor $\frac{|S|}{|T|}$ is then applied differently. GF1A seems to be specific to additive preferences while our intent is to define concepts that can conveniently be used for both additive and non-additive preferences.

Following this aim for a general definition that is suitable for general preference domains, we only consider groups of the same size to obtain ordinal properties. This is in the same spirit of envy-freeness and allows for  more generality.
It can also be argued that comparisons between same-sized groups implicitly captures comparisons between different sized groups: for a given $k$, one can compare the best subgroup in $S$ of size $k$ with the worst subgroup of $T$ of size $k$.
%the same size in the other group. 
% Moreover, we argue that considering only same-sized groups does not lose too much in generality. 
%Indeed \simon{To be completed...}

\medskip

Although GEF implies Pareto-optimality and envy-freeness, GEF1 is more stringent than the combination of the two criteria. The following example illustrates that even on very restricted preference domain these two concepts do not imply GEF1.

\begin{example}
	\label{ex:PO + EF1 =/> GEF1}
	Consider the following instance with eight items, from $o_1$ to $o_8$, and four agents, $a_1, a_2, a_3, a_4$ whose preferences are additive and single-peaked with respect to the axis $\tuple{o_1, \ldots, o_8}$. The utilities for the singletons are as follows.
	\begin{center}
		\setlength{\tabcolsep}{6pt}
		\begin{tabular}{c|cccccccc}
			& $o_1$ & $o_2$ & $o_3$ & $o_4$ & $o_5$ & $o_6$ & $o_7$ & $o_8$\\
			\hline
			$a_1$ & $-1$ & $-1$ & 1 & 1 & 0 & 0 & 0 & 0 \\
			$a_2$ & 0 & 0 & 0 & 0 & 1 & 1 & $-1$ & $-1$ \\
			$a_3$ & \squared{1} & \squared{1} & \squared{1} & \squared{1} & 0 & 0 & 0 & 0 \\
			$a_4$ & 0 & 0 & 0 & 0 & \squared{1} & \squared{1} & \squared{1} & \squared{1}
		\end{tabular}
	\end{center}
	We call $\pi$ the allocation represented by the squared items. $\pi$ is clearly envy-free: agents $a_3$ and $a_4$ have their maximal utility and their bundles give 0 utility to agents $a_1$ and $a_2$. The allocation is moreover Pareto-optimal. However, $S = \tuple{a_1, a_2}$, $T = \tuple{a_3, a_4}$, $\pi'_{a_1} = \{o_3, o_4, o_7, o_8\}$ and $\pi'_{a_2} = \{o_1, o_2, o_5, o_6\}$, are witnesses of a violation of GEF1. Agents in $S$ are better off with $\pi'$ than with $\pi$ even after removing one good: utilities after removal are $\tuple{1, 1}$ against $\tuple{0, 0}$.
\end{example}

When preferences are additive, it is well known (see \cite{ACI+18} for example) that proportionality is a relaxation of envy-freeness. In a similar spirit, one can define \emph{group-proportionality}, a relaxation of GEF that extends proportionality to groups. It corresponds to GEF when $T$ is fixed and set to $\agentSet$. Note that it corresponds to the \emph{core} as defined by \cite{FMS18}.

\begin{definition}[Group Proportionality]
	An allocation $\pi \in \Pi(\objSet, \agentSet)$ is group proportional (GP) if and only if for every $S \subseteq \agentSet$, there is no $\pi' \in \Pi(\objSet, S)$, such that $\forall i \in S, \frac{|S|}{n}u_i\left(\pi'_i\right) \geq u_i(\pi_i)$, with at least one strict inequality.
\end{definition}

We can then define the usual relaxations, GP1 and GPX. An allocation $\pi \in \Pi(\objSet, \agentSet)$ satisfies group proportionality up to one item (GP1) if for every $S \subseteq \agentSet$, $\pi' \in \Pi(\objSet, S)$, $i \in S$ there exists $O_i \subseteq \pi_i^- \cup \pi_i'^+$, $|O_i| \leq 1$, such that $\tuple{\frac{|S|}{n} u_i(\pi_i'\backslash O_i)}_{i \in S}$ does not Pareto-dominate $\tuple{u_i(\pi_i\backslash O_i)}_{i \in S}$. Group proportionality up to any item (GPX) can then be defined naturally.

\medskip

To conclude this section, we present in Figure \ref{fig:FainessCreiteriaLinks} a taxonomy of the different criteria discussed before. The links between s-GEF, GEF, GP and their relaxations are immediately derived from the definitions. Envy-freeness concepts are implied by GEF and s-GEF when $S$ and $T$ are singletons. GP implies PROP when $S$ is a singleton. PO is implied by s-GEF, GEF and GP for $S = T = \agentSet$.
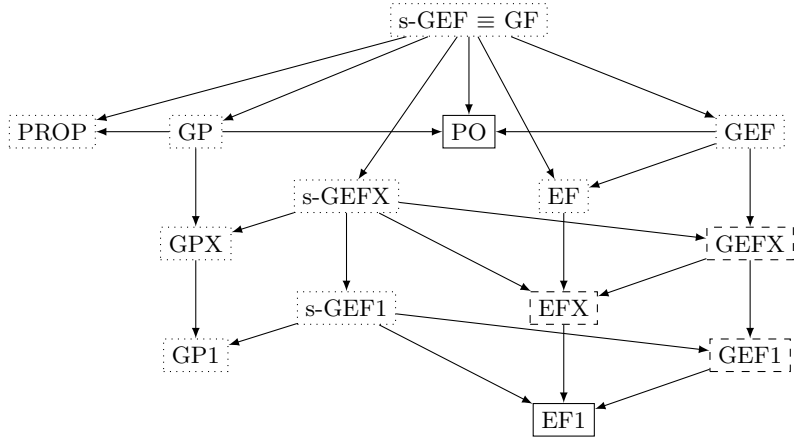
\begin{figure}[t]
	\centering
	\begin{tikzpicture}[node distance = 3.2em 	, > = latex]
	\node[draw, dotted] (s-GEF) {s-GEF $\equiv$ GF};
	
	\node[draw, below = of s-GEF] (PO) {PO};
	\node[below = 1.6em of PO] (tmpPO) {};
	
	\node[draw, right = 9em of PO, dotted] (GEF) {GEF};
	\node[draw, below = of GEF, dashed] (GEFX) {GEFX};
	\node[draw, below = of GEFX, dashed] (GEF1) {GEF1};
	
	\node[draw, left = 9em of PO, dotted] (GP) {GP};
	\node[draw, below = of GP, dotted] (GPX) {GPX};
	\node[draw, below = of GPX, dotted] (GP1) {GP1};
	
	\node[draw, left = 2.5em of tmpPO, dotted] (s-GEFX) {s-GEFX};
	\node[draw, below = of s-GEFX, dotted] (s-GEF1) {s-GEF1};
	
	\node[draw, right = 2.5em of tmpPO, dotted] (EF) {EF};
	\node[draw, below = of EF, dashed] (EFX) {EFX};
	\node[draw, below = of EFX] (EF1) {EF1};
	
	\node[draw, left = 3em of GP, dotted] (PROP) {PROP};
	
	%	\node[draw, right = 3em of GEFX, align = center] (lex) {Leximin\\optimality};
	
	\path[->] (s-GEF) edge [] (PO);
	\path[->] (s-GEF) edge [](GP);
	\path[->] (s-GEF) edge [] (PROP);
	\path[->] (s-GEF) edge [] (EF);
	
	\path[->] (GP) edge [] (PO);
	\path[->] (GP) edge [] (GPX);
	\path[->] (GPX) edge [] (GP1);
	
	\path[->] (EF) edge [] (EFX);
	\path[->] (EFX) edge [] (EF1);
	
	\path[->] (GP) edge [] (PROP);
	
	%	\path[->] (EFX) edge [] node [above, sloped] {Identical} (GEF1);
	
	\path[->] (s-GEF) edge [] (s-GEFX);
	\path[->] (s-GEFX) edge [] (s-GEF1);
	\path[->] (s-GEFX) edge [] (GPX);
	\path[->] (s-GEF1) edge [] (GP1);
	\path[->] (s-GEFX) edge [] (EFX);
	\path[->] (s-GEF1) edge [] (EF1);
	
	\path[->] (s-GEF) edge [] (GEF);
	\path[->] (s-GEFX) edge [] (GEFX);
	\path[->] (s-GEF1) edge [] (GEF1);
	
	\path[->] (GEF) edge [] (EF);
	\path[->] (GEFX) edge [] (EFX);
	\path[->] (GEF1) edge [] (EF1);
	\path[->] (GEF) edge [] (GEFX);
	\path[->] (GEFX) edge [] (GEF1);
	\path[->] (GEF) edge [] (PO);
	
	%	\path[->] (lex) edge [] node [below, sloped] {Ternary} (GEF1);
	\end{tikzpicture}
	\caption{Logical relationship between fairness and efficiency criteria. Allocations satisfying concepts in dotted (resp. plain) are not (resp. are) guaranteed to exist. No existence  result is known for concepts in dashed.}
	
	\label{fig:FainessCreiteriaLinks}
\end{figure}

\section{The Egal-Sequential Algorithm for Identical Utilities}

In this section, we present the Egal-Sequential Algorithm that returns a GEF1 allocation when preferences are identical. 
Identical preferences constitute an important and natural class of preferences especially if the item's values are objective or publicly known. The algorithm allocates sequentially the items in decreasing order of absolute utility. The item to be allocated is given to the worse off agent if it is a good and to the better off agent otherwise.
We prove that this ensures GEF1 in two steps. We first show that the Egal-Sequential Algorithm returns an EFX allocation and then that every EFX allocation is also GEF1 when preferences are identical. Preferences are identical if there exists a common utility function $u$ shared by all the agents: $\forall i \in \agentSet, u_i = u$.

% \begin{lemma}
% 	Let $I = \tuple{\agentSet, \objSet, (u_i)_{i \in \agentSet})} \in \instBoth$ be an instance, if $(u_i)_{i \in \agentSet}$ is a profile of identical utility functions then the Egal-Sequential Algorithm returns an allocation $\pi$ that satisfies EFX in time in $\mathcal{O}(mn)$.
% \end{lemma}

\begin{lemma}
	\label{lemma:egalSeqAlgoCorrect}
	% Let $I = \tuple{\agentSet, \objSet, (u_i)_{i \in \agentSet})} \in \instBoth$ be an instance consisting of identical utilities.
	For identical utilities, the Egal-Sequential Algorithm returns an allocation $\pi$ satisfying EFX and takes time $\mathcal{O}(\max\{m \log m, mn\})$.
\end{lemma}

\begin{proof}[Sketch of the proof]
	We show that throughout an execution of the algorithm, the partial allocation $\pi$ is always EFX. If a chore $o$ it to be allocated, it is given to the best off agent $i^*$. The only possible increase in the envy is for $i^*$ to start envying some agents. However, since preferences are identical, $i^*$ did not envy anyone before getting $o$. Hence, removing $o$ from $\pi_{i^*}$ eliminates $i^*$'s envy. Since $o$ is the smallest chore, $\pi$ is EFX. The case when $o$ is a chore is symmetric.
\end{proof}

\begin{algorithm}[t]
	\DontPrintSemicolon
	\KwIn{An instance $I = \tuple{\agentSet, \objSet, (u_i)_{i \in \agentSet})}$ with identical utility function $u$}
	\KwOut{$\pi \in \Pi(\objSet, \agentSet)$ an allocation satisfying EFX and GEF1}
	Set $\pi$ to the empty allocation\;
	Order items $o_1,\ldots, o_m$ in $\objSet$ in decreasing order of $|u(o)|$\;
	\For{$j=1$ to $m$}{
		\If{$u(o_j) \geq 0$}{
			Choose $i^* \in \arg \min_{i \in \agentSet} u(\pi_i)$ \;
		}
		\Else{Choose $i^* \in \arg \max_{i \in \agentSet} u(\pi_i)$ \;}
		Give $o_j$ to $i^*$: $\pi_{i^*} \gets \pi_{i^*} \cup \{o_j\}$ \;
	}
	\Return $\pi$ \;
	\caption{The Egal-Sequential Algorithm}
	\label{algo:egalgreedy}
\end{algorithm}

\begin{lemma}
	\label{lemma:EFX=>GEF1Identical}
	Under identical utilities, any allocation satisfying EFX is GEF1.
\end{lemma}	

\begin{proof}[Sketch of the proof]
	Consider an allocation $\pi$ satisfying EFX but not GEF1. There exist then $S \subseteq \agentSet$, $T \subseteq \agentSet$ and $\pi' \in \Pi(\pi_T, S)$ that are witnesses of a violation of GEF1. To get an increase up to one item, agents in $S$ should receive at least $s_i = \max \{ \max_{o \in \pi_i'^+} u(o), \max_{o \in \pi_i^-} -u(o) \}$ additional utility.
	Violating GEF1 then implies: $\forall i \in S, u(\pi'_i) - u(\pi_i) \geq s_i$ with one strict inequality. 
%	The 0 component in the definition of $s_i$ is meant to tackle the case when an agent receives an empty allocation in $\pi$ and no goods in $\pi'$. If $\pi_i = \pi_i'^+ = \emptyset$ we should have $\pi_i^- = \emptyset$ to get a GEF1 violation. Hence, we have $u(\pi_i) = u(\pi_i')$ and $s_i = 0$ is a suitable bound.
	By summing these over $i \in S$, 
%	we obtain $\sum_{i \in S} u(\pi_i') - \sum_{i \in S} u(\pi_i) > \sum_{i \in S} s_i$. 
	and since utilities are identical and additive, we have:
	\begin{align}
	\label{eq:proofEFX=>GEF1Line1}
	u(\pi_T) - u(\pi_S) > \sum_{i \in S} s_i.
	\end{align}
	
	As $\pi$ is EFX, the difference in utility between two agents cannot exceed $s_{i, j} = -\min\{\min_{o \in \pi_j^+} -u(o), \min_{o \in \pi_i^-} u(o)\}$.
%	we have $\forall i, j \in \agentSet, \forall o \in \pi_i^- \cup \pi_j^+, u(\pi_i \backslash \{o\}) \geq u(\pi_j \backslash \{o\})$. For $i, j \in \agentSet$, we define $s_{i, j}$ as: $s_{i, j} = -\min\{\min_{o \in \pi_j^+} -u(o), \min_{o \in \pi_i^-} u(o), 0\}$.
%	The 0 component is once again here to take care of the case when $i$ receives an empty allocation and $j$ has only chores. 
	We have thus:
	\begin{align}
	\label{eq:proofEFX=>GEF1Line2}
	\forall i, j \in \agentSet, s_{i, j} \geq u(\pi_j) - u(\pi_i).
	\end{align}
	
	Our goal is now to sum up inequalities describe in \eqref{eq:proofEFX=>GEF1Line2} to obtain a contradiction with \eqref{eq:proofEFX=>GEF1Line1}. We are thus looking for pairs $(i, j)$ such that each $i \in S$ and each $j \in T$ appears once and only once, and such that the sum of $s_{i, j}$ over these pairs is smaller than $\sum_{i \in S} s_i$.
	
	\medskip
	
	Let us consider the bipartite graph $G = \tuple{S \cup T, E}$ where nodes represent agents in $S$ and in $T$. There is an edge $(i, j) \in E$ between agents $i \in S$ and $j \in T$ if and only if $\pi_i'^+ \cap \pi_j\neq \emptyset$, that is, $i$ receives some of $j$'s goods in $\pi'$. 
	
	We consider a partition of $S$ into $S^+ = \{ i \in S \mid u(\pi_i) \geq 0 \}$ and $S^- = S \backslash S^+$. Using Hall's theorem \citep{Hall35} we show that there always exists a matching $M \subseteq S \times T$ in $G$ that matches all the agents in $S^+$. The proof is omitted.

	We then extend the matching $M$ to match all the agents in $S$ by arbitrarily pairing each agent $i\in S^-$ with an unmatched agent in $T$. Let the extended matching be called $M^*$. Observe that for any new pair of agents $(i, k) \in M^* \backslash M$, we have $i \in S^-$, that is $\pi_i^- \neq \emptyset$, hence for any agent $j \in T$, we have:
	\begin{align}
	\label{eq:proofEFX=>GEF1Line6}
	s_{i, j} \leq \min_{o \in \pi_i^-} u(o) \leq \max_{o \in \pi_i^-} u(o) \leq s_i.
	\end{align}
	Moreover, for any pair of agents $(i, j) \in E$, we have $\pi_i'^+ \cap \pi_j\neq \emptyset$, hence:
	\begin{align}
	\label{eq:proofEFX=>GEF1Line5}
	s_i \geq \max_{o \in \pi_i'^+} u(o) \geq \max_{o \in \pi_i'^+ \cap \pi_j} u(o) \geq \min_{o \in \pi_j^+} u(o) \geq s_{i, j}.
	\end{align}
	From \eqref{eq:proofEFX=>GEF1Line5} and \eqref{eq:proofEFX=>GEF1Line6} we get that $\sum_{(i, j) \in M^*} s_{i, j} \leq \sum_{i \in S} s_i$. Hence, summing \eqref{eq:proofEFX=>GEF1Line2} over $(i, j) \in M^*$ brings $\sum_{i \in S} s_i \geq \sum_{(i, j) \in M^*} s_{i, j} \geq u(\pi_T) - u(\pi_S)$.  This contradicts \eqref{eq:proofEFX=>GEF1Line1}, we have thus proved that $\pi$ satisfies both EFX and GEF1.
\end{proof}
	
A direct consequence of the two previous lemmas is that a GEF1 allocation can be computed by the Egal-Sequential Algorithm.

\begin{theorem}
	\label{thm:existGEF1Identical}
	For identical utilities, an allocation satisfying GEF1 always exists and can be computed in linear time by the Egal-Sequential Algorithm.
\end{theorem}

\citet{CFS+19} showed that when preferences are identical their relaxation of group-fairness is implied by EFX. We significantly extend their result in different ways. Firstly, our result applies in the case of mixed utilities where preferences can model both goods and chores. Secondly, we provide a linear time algorithm to compute GEF1 allocation with identical preferences. Finally, our proof does not involve the Nash social welfare which is not a suitable solution concept with chores.\footnote{When considering only chores, either maximizing or minimizing the absolute Nash social welfare does not imply EF1 for instance.}
%Consider four items and two agents with the following additive utilities: $u_1(o_1) = -1$, $u_1(o_2) = u_1(o_3) = u_1(o_4) = -100$, $u_2(o_1) = -6$ and $u_2(o_2) = u_2(o_3) = u_2(o_4) = -2$. Maximizing the absolute value of the Nash welfare leads to the allocation $\tuple{\{o_2, o_3, o_4\}, \{o_1\}}$ in which $a_1$ envies $a_2$ up to one chore, similar example can be states for the minimization of the Nash social welfare.}.

%One can wonder whether a similar result can be stated for GEFX. Although we do not have a counter-example for it, the current proof cannot be used for that matter as $s_i$ are defined as maxima, while for GEFX it should be a minima which would not lead to the contradiction we have settle her.

%It is moreover clear that GEF1 does no implies EFX. Take for a simple counter-example two agents whose preferences over 3 goods are $u(o_1) = 1$, $u(o_2) = 2$ and $u(o_3) = 3$. The allocation $\pi = \tuple{\{o_1, o_3\}, \{o_2\}}$ is GEF1 but not EFX.

%Remember that in Example \ref{ex:NonExistences} and we showed that no allocation was s-GEF1, this implies that Theorem \ref{thm:existGEF1Identical} can no be extended to s-GEF1 even with identical preferences.

\section{The Ternary Flow Algorithm}

In this section, we focus on another restriction of the preferences, namely \emph{ternary symmetric preferences}. We say that agent $a_i$ has ternary symmetric preferences if her preferences are additive and the utilities of the singletons are taken from the set $\{-\alpha_i, 0, \alpha_i\}$ for a given $\alpha_i > 0$. 

We provide an algorithm that computes GEF1 allocations for ternary symmetric preferences. We do so by proving that any leximin-optimal allocation is also GEF1 and by providing an algorithm returning a leximin-optimal allocation in polynomial time. % Previously, links between leximin-optimality and envy-freeness concepts have been observed by \citet{PlRo18} for EFX.
Similar links between leximin-optimality and envy-freeness concepts have been observed by \citet{PlRo18} for the case of goods.

We first provide a characterization of Pareto-optimality for ternary symmetric utilities.

\begin{lemma}
	\label{prop:characPOTernaryUtil}
	Let $\instBoth = \tuple{\agentSet, \objSet, (u_i)_{i\in \agentSet}}$ be an instance where $(u_i)_{i \in \agentSet}$ describes a profile of ternary symmetric utility functions. An allocation $\pi \in \Pi(\agentSet, \objSet)$ is Pareto-optimal if and only if for every item $o \in \objSet$ we have:
	$$\left\{\begin{array}{l}
	o \in \pi_i \text{ with } u_i(o) > 0, \text{iff } \max_{j \in \agentSet} u_j(o) > 0, \\
	o \in \pi_i \text{ with } u_i(o) = 0, \text{iff } \max_{j \in \agentSet} u_j(o) = 0, \\
	o \in \pi_i \text{ with } u_i(o) < 0, \text{iff } \max_{j \in \agentSet} u_j(o) < 0.
	\end{array}\right.$$
\end{lemma}

\begin{proof}
	Since Pareto-optimality is invariant under rescaling of utilities, assume w.l.o.g. that $\alpha_i = 1$, $\forall i \in \agentSet$. Then, if the three conditions hold, every item is allocated to an agent having maximal utility for it. This implies Pareto-optimality.
%	that is the items are allocated an agent that has highest utility for it. This ensure that the maximal utilitarian welfare is achieved, which implies Pareto-optimality.
	
	Next, assume that one item $o$ is allocated to an agent $i$ who do not have maximal utility for it: $o$ is a chore (resp. neutral) for $i$ and there is $j$ who considers $o$ as either neutral or a good (resp. a good). Then transferring $o$ from $i$ to $j$ leads to a Pareto-improvement.
%	, such that $u_i(\pi_i \backslash\{o\}) \geq u_i(\pi_i)$ and $\exists a_j \in \agentSet$ with $u_j(o) > u_i(o)$. Thus transferring $o$ from $a_i$ to $a_j$ leads to a Pareto-improvement, hence $\pi$ is not Pareto-optimal.
\end{proof}

For a profile of ternary symmetric preferences $(u_i)_{i \in \agentSet}$, we introduce the normalized profile$(u_i^{\text{Norm}})_{i \in \agentSet}$ that corresponds to $(u_i)_{i \in \agentSet}$ but such that every $\alpha_i$ as been set to 1. For a normalized profile, every singleton has utility in $\{-1, 0, 1\}$.

This preference domain models statements such as ``I like'', ``I am indifferent'' and ``I do not like''. It is close to the idea of approval and disapproval voting \citep{BrFi78, Fels89}. In the fair division literature, it has also been referred to as dichotomous preferences when there are only goods \citep{BMS05}.

\medskip

Next we introduce leximin optimality. For an allocation $\pi$ we denote by $\vec{u}(\pi) \in \mathbb{R}^n$ the vector of the utilities in $\pi$ sorted increasingly.
%, the $i$-th component of $\vec{u}(\pi)$, noted $\vec{u}_i(\pi)$, is then the utility of the agent who has the $i$-th lowest utility in $\pi$. 
For two vectors $\vec{u}, \vec{v} \in \mathbb{R}^k$, we say that $\vec{u}$ leximin-dominates $\vec{v}$, written $\vec{u} \succ_{lex} \vec{v}$, if there exits an index $i \leq k$ such that $\vec{u}_j = \vec{v}_j, \forall j < i$, and $\vec{u}_i > \vec{v}_i$. Finally, an allocation $\pi$ is said to be leximin-optimal if there is no allocation $\pi'$ such that $\vec{u}(\pi) \succ_{lex} \vec{u}(\pi')$. 

\medskip

%The algorithm we use to compute leximin-optimal allocations uses network flows. For details on network flows, we refer the reader to \citet[chapters 10 and 12]{Schr03}. 
% A cost flow network is represented by a directed graph $G = \tuple{V, E}$ whose nodes $v \in V$ are labelled with a demand $d(v) \in \mathbb{R}$ while edges $e \in E$ are labelled by a maximum capacity $\delta(e) \in \mathbb{R}_{> 0}$ and a cost $c(e) \in \mathbb{R}_{\geq 0}$. A flow $f$ is a mapping $f: E \rightarrow \mathbb{R}_{\geq 0}$ where $f(e)$ is the amount of flow passing through the edge $e$. By a slight abuse of notation, we use $f(v)$ for a node $v \in V$ to denote the difference between the outgoing flow and the incoming flow in $v$. A flow is realizable if for every edge $e \in E$ its capacity is not exceeded, $f(e) \leq \delta(e)$, and for every node $v \in V$ its demand is satiated, $f(v) = d(v)$. The cost of a realizable flow $f$ is defined by $c(f) = \sum_{e \in E} f(e) c(e)$, a minimum cost flow is then a flow with minimal cost among the realizable flows. A flow is said to be integer if $\forall e \in E, f(e) \in \mathbb{N}$. 
%
%\medskip
%
Our algorithm makes used of the Nash Flow Algorithm presented by \citet{DaSc15}. It computes in polynomial time an allocation maximizing the Nash social welfare when preferences are binary, i.e., the singletons' utilities are in $\{0, 1\}$, by using a cost flow network \citep[chapters 10 and 12]{Schr03}.\footnote{\citet{BKV18} also proposed a polynomial-time algorithm for maximizing the Nash social welfare with binary utilities. However, it does not imply leximin-optimality which we need to achieve GEF1.} It constructs a specific network for which any minimum integer cost flow corresponds to an allocation maximizing the Nash social welfare.

\medskip

We first extend \cite{DaSc15}'s result by showing that the allocation returned by the Nash Flow Algorithm is also leximin-optimal.

\begin{lemma}
	\label{lemma:TernNashFlow=>Lex}
	Let $I = \tuple{\agentSet, \objSet, (u_i)_{i\in\agentSet}} \in \instGood$ be an instance with only goods and where $(u_i)_{i \in \agentSet}$ describes a profile of normalized ternary utility functions. An allocation $\pi$ is leximin-optimal if and only if it correspond to a minimum cost integer flow in the network defined by the Nash Flow Algorithm.
\end{lemma}

This equivalence implies that any leximin-optimal allocation also maximizes the Nash social welfare.

\begin{corollary}
	Let $I = \tuple{\agentSet, \objSet, (u_i)_{i\in\agentSet}} \in \instGood$ be an instance with only goods and where $(u_i)_{i \in \agentSet}$ describes a profile of normalized ternary utility functions. Any leximin-optimal allocation $\pi$ maximizes the Nash social welfare.
\end{corollary}

Making use of the Nash Flow Algorithm, we propose the Ternary Flow Algorithm (Algorithm \ref{algo:lexminTern}). It computes a leximin-optimal allocation on the normalized utilities which corresponds to a GEF1 allocation w.r.t. the original preferences.

\begin{algorithm}[t]
	\DontPrintSemicolon
	\KwIn{An instance $I = \tuple{\agentSet, \objSet, (u_i)_{i \in \agentSet})}$ such that $\forall i \in \agentSet, u_i = u$, for a given utility function $u$}
	\KwOut{$\pi \in \Pi(\objSet, \agentSet)$ a GEF1 allocation}
	%	Consider new utility functions $(\overline{u_i})_{i \in \agentSet}$ such that $\forall i \in \agentSet, \forall o \in \objSet, \overline{u_i}(o) = \left\{\begin{array}{cl}
	%	1 & \text{if } u_i(o) > 0\\
	%	0 & \text{if } u_i(o) = 0\\
	%	-1 & \text{if } u_i(o) < 0\\
	%	\end{array}\right.$\;
	Set $O^+$ to $\{o\in \objSet: \max_{i\in \agentSet} u_i(o)>0\}$\;
	Set $O^0$ to $\{o\in \objSet: \max_{i\in \agentSet} u_i(o)=0\}$\;
	Set $O^-$ to $\{o\in \objSet: \max_{i\in \agentSet} u_i(o)<0\}$\;
	Consider new utility functions $(u_i')_{i \in \agentSet}$ such that $\forall i \in \agentSet, \forall o \in O^+, u_i'(o) = \left\{\begin{array}{cl}
	1 & \text{if } u_i^{\text{Norm}}(o) = 1, \\ 
	0 & \text{otherwise }
	\end{array}\right.$ \label{algoLine:lexminTernUtilities} \;
	Run the Nash Flow Algorithm on $I' = \tuple{\agentSet, O^+, (u_i')_{i \in \agentSet}}$ to obtain the partial allocation $\pi$ \label{algoLine:lexminTernNashFlow}\;
	\For{$o \in O^-$ \label{algoLine:lexminTernFor1}}{
		Allocate $o$ to $i^* \in \arg\max_{i \in \agentSet} u_i^{\text{Norm}}(\pi_i)$ and update $\pi$\;}
	\For{$o \in O^0$}{
		Allocate $o$ to $i^* \in \arg\min_{i \in \agentSet, u_i^{\text{Norm}}(o) = 0} u_i^{\text{Norm}}(\pi_i)$ and update $\pi$ \;}
	\Return $\pi$ \;
	\caption{The Ternary Flow Algorithm}
	\label{algo:lexminTern}
\end{algorithm}

\begin{lemma}
	\label{lemma:algoLeminTernCorrect}
	For ternary symmetric preferences, the Ternary Flow Algorithm returns allocations that are leximin-optimal for the normalized preferences.
\end{lemma}

\begin{proof}[Sketch of the proof]
	From Lemma \ref{lemma:TernNashFlow=>Lex} we know that $\pi$, returned by the Nash Flow Algorithm, is leximin-optimal. We claim that allocating $o \in O^-$, which is a chore for everyone, to the best off agent maintains leximin-optimality. Finally, giving items in $O^0$ to agents who value them 0 also preserves leximin-optimality.
\end{proof}

\begin{lemma}
	\label{lemma:lexmin=>GEF1Tern}
	If the preferences are normalised ternary symmetric, any leximin-optimal allocation also satisfies GEF1.
\end{lemma}

\begin{proof}[Sketch of the proof]
	Let $\pi$ be a leximin-optimal allocation. The proof is based on the two following claims derived from leximin-optimality. Proofs are omitted.
		
	\begin{claim}
		\label{claim:lexmin=>GEF1Claim2}
		If $min_{i \in \agentSet}u_i(\pi_i) < 0$, then $\forall j \in \agentSet, 0 \leq u_j(\pi_j) - min_{i \in \agentSet}u_i(\pi_i) \leq 1$. 
	\end{claim}
	%	\begin{claimproof}
	%		Let $k \in \agentSet$ be the agent with minimum utility in $\pi$. Assume that $u_k(\pi_k) < 0$, it means that $k$ owns at least one chore $o$ in $\pi$. Since $\pi$ is Pareto-efficient, from Lemma \ref{prop:characPOTernaryUtil} we know that $o$ is a chore for every agent in $\agentSet$. Let us assume that there exists an agent $j \in \agentSet$ such that $u_j(\pi_j) > u_k(\pi_k) + 1$, then since $o$ is also a chore for $j$, the allocation obtained by transferring $o$ from $k$ to $j$ would leximin-dominates $\pi$, which violates the assumption of leximin-optimality of $\pi$. 
	%	\end{claimproof}
	
	\begin{claim}
		\label{claim:lexmin=>GEF1Claim1} For every $i, j \in \agentSet$, if $u_j(\pi_j) - u_i(\pi_i) \geq 2$ then $\{o \in \pi_j \mid u_i(o) = 1\} \cup \{o \in \pi_i \mid u_j(o) = -1\} = \emptyset$.
	\end{claim}
%	\begin{claimproof}
%		Indeed assume that there is an item $o \in \pi_j$ such that $u_i(o) = 1$, then consider the allocation $\pi'$ such that $\forall k \in \agentSet \backslash \{i, j\}, \pi'_k = \pi_k$, $\pi'_i = \pi_i \cup \{o\}$, $\pi'_j = \pi_j \backslash \{o\}$. Since $u_j(\pi_j) - u_i(\pi_i) \geq 2$, $i$ appears before $j$ in $\vec{u}(\pi)$, moreover, as $u_j(\pi_j') - u_i(\pi_i') \geq 0$, $i$ also appears before $j$ in $\vec{u}(\pi')$. However, $u_i(\pi'_i) > u_i(\pi_i)$ implies that $\vec{u}(\pi') \succ_{lex} \vec{u}(\pi)$ which contradicts the fact that $\pi$ is leximin-optimal. The case where there exists an item $o \in \pi_i$ such that $u_j(o) = -1$ is exactly symmetric. 
%	\end{claimproof}
	
	Assume toward a contradiction that there exist $S$, $T$ and $\pi'$ which are witnesses of a violation of GEF1. For any agent in $S$, her utility in $\pi'$ should be at least one higher than in $\pi$, and there exists and agent in $S$ who increases her utility by at least two. The first claim implies that this is not possible if some agent receives negative utility. The second claim states that each agent in $T$ can increase by at most one the utility of an agent in $S$. The violation of GEF1 is therefore not possible and we have proved that $\pi$ is both leximin-optimal and GEF1 for the normalized preferences.
\end{proof}

From Lemma \ref{lemma:algoLeminTernCorrect} and \ref{lemma:lexmin=>GEF1Tern}, we derive the statement for GEF1 allocations.

\begin{theorem}
	\label{thm:existsGEF1Tern}
	Let $I = \tuple{\agentSet, \objSet, (u_i)_{i\in\agentSet}} \in \instBoth$ be an instance where $(u_i)_{i \in \agentSet}$ describes a profile of ternary symmetric utility functions. A GEF1 allocation always exists and can be computed in polynomial time via the Ternary Flow Algorithm.
\end{theorem}

\begin{proof}
	Let $\pi$ be the allocation returned by the Ternary Flow Algorithm. By Lemma \ref{lemma:algoLeminTernCorrect}, $\pi$ is leximin-optimal for normalized preferences and thus GEF1 for these preferences (Lemma \ref{lemma:lexmin=>GEF1Tern}). Since no interpersonal comparison are required for GEF1, any allocation satisfying it for normalized preferences also does for non-normalized preferences. Allocation $\pi$ therefore satisfies GEF1.
\end{proof}

%	\subsection{Leximin-optimal allocation and GEF1 in the general case}
%	
%	In the light of theorems \ref{thm:existGEF1Identical} and \ref{thm:existsGEF1Tern} it may seems like the leximin social welfare is closely linked to GEF1. We show in the following example that this is not true for more general preferences. \cite{CKM+16} proved that for additive preferences over goods, the leximin solution does not always satisfies EF1, we provide similar example for chores.
%	
%	\begin{example}
%		Consider the following instance with four items from $o_1$ to $o_4$ and trhee agents $a_1, a_2$ and $a_3$ whose preferences are additive such that the weights for the singleton are as follows.
%		\begin{center}
%			\begin{tabular}{c|cccc}
%				& $o_1$ & $o_2$ & $o_3$ & $o_4$ \\
%				\hline
%				$a_1$ & \squared{0} & $-1/3$ & $-1/3$ & $-1/3$ \\
%				$a_2$ & $-5/7$ & \squared{0} & \squared{$-1/7$} & \squared{$-1/7$} \\
%				$a_3$ & 0 & 0 & $-1/2$ & $-1/2$
%			\end{tabular}
%		\end{center}
%		We call $\pi = \tuple{\{o_1\}, \{o_2, o_3, o_4\}, \emptyset}$ the squared allocation. This allocation is leximin-optimal, however agent $a_2$ envies both agent $a_1$ and $a_3$, it is thus not EF1 hence not GEF1.
%	\end{example}

\section{Non-existence of fair allocations}

In this section, we present negative existence results for many GEF-related concepts. We show in particular that as soon as we allow for groups of different size, existence of fair allocations cannot be guaranteed.

%One of the main results presented by \citet{CFS+19} is that both of their relaxations of group-fairness are implied by local Nash optimality, a relaxation of the Nash social welfare, when there are only goods and preferences are additive. This is particularly interesting as it proves the existence of group-fair up to one good allocations and provides a pseudo polynomial time algorithm to compute such allocations. Let $I = \tuple{\agentSet, \objSet, (u_i)_{i \in \agentSet}} \in \instGood$ be an instance with only goods, an allocation $\pi \in \Pi(\objSet, \agentSet)$ is said to be locally Nash optimal if one cannot improve the Nash welfare by simply transferring one good from one agent to another: $\forall i, j \in \agentSet, \forall o \in \pi_j, u_j(o) > 0 \text{ and } u_i(\pi_i) \times u_j(\pi_j) \geq u_i(\pi_i\cup\{o\}) \times u_j(\pi_j\backslash\{o\})$.
%
%As GEF1 is equivalent to \citeauthor{CFS+19}'s relaxations when considering group of the same size, when there are only goods and preferences are additive, any locally Nash optimal allocations satisfies \GFoGood{}. We show that it is no longer true when considering groups of different size, and even more that there are no guarantee of existence of s-GEF1 allocations even for only goods and additive preferences.

\begin{example}
	\label{ex:NonExistences}
	Let us consider an instance with three agents, $a_1, a_2$ and $a_3$, and three goods, $o_1, o_2$ and $o_3$, where preferences are additive and defined as follows:
	\begin{center}
		\setlength{\tabcolsep}{6pt}
		\begin{tabular}{c|ccc}
			& $o_1$ & $o_2$ & $o_3$ \\
			\hline
			$a_1$ & \squared{1} & 1 & $\epsilon$ \\
			$a_2$ & 1 & \squared{1} & $\epsilon$ \\
			$a_3$ & 1 & 1 &\squared{$\epsilon$}
		\end{tabular}
	\end{center}
	with $0 < \epsilon < 1/3$. We call $\pi$ the allocation defined by the squared items. First, observe that $S = \{a_3\}$, $T = \{a_1, a_2, a_3\}$ and the reallocation $\pi' = \tuple{\{o_1, o_2, o_3\}}$  are witnesses of the violation of s-GEF1 in $\pi$:
	$$\frac{|S|}{|T|}u_3(\pi'_3\backslash\{o_1\}) = \frac{1}{3}u_3(\{o_2, o_3\}) = \frac{1}{3} (1 + \epsilon) > u_3(\pi_3) = \epsilon.$$
	
	One can see that no allocation satisfies s-GEF1 in this example. Indeed, if an agent receives more than one item then, another one would envy her up to one item. There are therefore no guarantees of existence for allocations satisfying s-GEFX and s-GEF. One can moreover see that $\pi$ is not GP1, hence existence GPX and GP allocation cannot be guaranteed.
	
	Another interesting observation is that $\pi$ maximizes the Nash social welfare but is not s-GEF1.
\end{example}

Next, we show that while the existence of GEF1 allocations is guaranteed when there are only goods with additive preferences, it is no longer the case for monotonic preferences.

\begin{example}
	Let us consider two agents, $a_1$ and $a_2$, whose preferences, $u_1$ and $u_2$, are presented below. The preferences only depend on the number of items received by each agent. Agent $a_2$ gets positive utility only if she receives at least 3 items.
	\begin{center}
		\setlength{\tabcolsep}{6pt}
		\begin{tabular}{c|c|c}
			$X \subseteq \objSet$ & $u_1(X)$ & $u_2(X)$ \\
			\hline
			$|X| = 4$ & 10 & 10 \\
			$|X| = 3$ & 6 & 6 \\
			$|X| = 2$ & 4 & 0 \\
			$|X| = 1$ & 1 & 0 \\
			$|X| = 0$ & 0 & 0
		\end{tabular}
	\end{center}
	In such instance, the only allocations satisfying GEF1 for $S = T = \mathcal{N}$ are the ones in which $a_1$ gets either all the goods, none or exactly one. None of these allocations are EF1, hence no allocation is GEF1.
%	: if $a_1$ receives zero or one good, she is envious of $a_2$ and if she receives all the goods, $a_2$ envies her.
\end{example}

% One can wonder why is the case with only chores not symmetric as the one with only goods. In a very basic example, we show that maximizing of minimizing the Nash welfare could lead to an allocation that is not EF1 when there are only chores.
%
% \begin{example}
% 	Let us consider an instance with two agents $a_1$ and $a_2$ and four chores $o_1, o_2, o_3$ and $o_4$ where preferences are additive and defined as follow:
% 	\begin{center}
% 		\begin{tabular}{c|cccc}
% 			& $o_1$ & $o_2$ & $o_3$ & $o_4$\\
% 			\hline
% 			$a_1$ & -1 & -100 & -100 & -100 \\
% 			$a_2$ & -6 & -2 & -2 & -2
% 		\end{tabular}
% 	\end{center}
% 	The allocation maximizing the absolute value of the Nash welfare is $\tuple{\{o_2, o_3, o_4\}, \{o_1\}}$ in which $a_1$ envies $a_2$ up to one chore. Let us now consider the following instance:
% 	\begin{center}
% 		\begin{tabular}{c|cccc}
% 			& $o_1$ & $o_2$ & $o_3$ & $o_4$\\
% 			\hline
% 			$a_1$ & -1 & -100 & -100 & -100 \\
% 			$a_2$ & -1 & -1 & -1 & -1
% 		\end{tabular}
% 	\end{center}
% 	Here, minimizing the Nash welfare leads to $\tuple{\{o_1\}, \{o_2, o_3, o_4\}}$ where $a_2$ envies $a_1$ up to one chore.
% \end{example}

\section{Testing GEF1 is coNP-complete}

%\citet{CFS+19} proved that for goods, checking whether an allocation satisfies GF1A or GF1B is coNP-complete. 
%Even if the concepts are slightly different, the same result can be shown for our concepts.
%We prove that the hardness even holds when we consider concepts that compare sets of equal sizes. 
We prove in this section that testing GEF1 is coNP-complete when there are only goods, only chores and both of them. The decision problem is the following.
\begin{center}
	\begin{tabular}{rl}
		\toprule
		\multicolumn{2}{c}{\textsc{is-GEF1}}\\
		\midrule
		\textbf{Instance:} & An instance $I = \tuple{\agentSet, \objSet, (u_i)_{i \in \agentSet}} \in \instBoth$ and $\pi \in \Pi(\objSet, \agentSet)$.\\
		\textbf{Question:} & Does $\pi$ satisfy GEF1?\\
		\bottomrule
	\end{tabular}
\end{center}
We use \textsc{is-\GFoGood{}} and \textsc{is-\GFoChore{}} to refer to the same decision problem when there are respectively only goods ($I \in \instGood$) and only chores $(I \in \instChore$).

\begin{theorem}
	\label{thm:complexityTestingGEF1}
	The problems \textsc{is-GEF1}, \textsc{is-\GFoGood{}} and \textsc{is-\GFoChore{}} are strongly coNP-complete.
\end{theorem}

\begin{proof}[Sketch of the proof]
	We present the reduction for the \textsc{is-\GFoChore{}} problem. By reducing the \textsc{3-Partition} problem \citep{GaJo75}, we show that checking if $\pi$ violates GEF1 when there are only chores is strongly NP-complete.
	\begin{center}
		\begin{tabular}{rl}
			\toprule
			\multicolumn{2}{c}{\textsc{3-Partition}}\\
			\midrule
			\textbf{Instance:} & A multi-set of $3m$ numbers $X = \{x_1, \ldots, x_{3m}\}$ such that:\\
			& $\forall x \in X, 1/4 < x < 1/2$ and $\sum_{x \in X} x = m$.\\
			\textbf{Question:} & Is there a partition $(X_i)_{i \in \llbracket 1, m \rrbracket}$ of $X$ such that $\forall i, \sum_{x \in X_i} = 1$ ? \\
			\bottomrule
		\end{tabular}
	\end{center}
	
	Let $X = \{x_1, \ldots, x_{3m}\}$ be an instance of the \textsc{3-Partition} problem. We present in the following its corresponding instance $(I, \pi)$ of the \textsc{is-\GFoChore{}} problem. The set of chores is $\objSet = \{g_1, \ldots, g_m\} \cup \{h_1, \ldots, h_m\} \cup \{l_1, \ldots, l_{3m}\} \cup \{o_1, \ldots, o_{2m}\}$ and the set of agents $\agentSet = \{a_1, \ldots, a_m\} \cup \{b_1, \ldots, b_m\}$. The utilities of the singletons are as follows:
	\begin{center}
		\vspace{-0.4cm}
		\resizebox{\linewidth}{!}{
			$\begin{array}{c|ccc|ccc|ccccccc|cccccc}
			& g_1 & \ldots & g_m & h_1 & \ldots & h_m & l_1 & l_2 & l_3 & \ldots & l_{3m - 2} & l_{3m - 1} & l_{3m} & o_1 & \ldots & o_m & o_{m + 1} & \ldots & o_{2m} \\
			\hline
			a_1 & \boxed{-m - \epsilon} & -M & -M & \boxed{-1 - \epsilon} & -M & 0 & -x_1 & -x_2 & -x_3 & \cdots & -x_{3m - 2} & -x_{3m - 1} & -x_{3m} & \boxed{0} & -M & -M & -M & -M & -M  \\
			\vdots & -M & \ddots & -M & -M & \ddots & -M & -x_1 & -x_2 & -x_3 & \cdots & -x_{3m - 2} & -x_{3m - 1} & -x_{3m} & -M & \ddots & -M & -M & -M & -M \\
			a_m & -M & -M & \boxed{-m - \epsilon} & -M & -M & \boxed{-1 - \epsilon} & -x_1 & -x_2 & -x_3 & \cdots & -x_{3m - 2} & -x_{3m - 1} & -x_{3m} & -M & -M & \boxed{0} & -M & -M & -M  \\
			\hline
			b_1 & -x_2 - x_3 & -M & -M & -M & -M & -M & \boxed{-x_1} & \boxed{-x_2} & \boxed{-x_3} & -M & -M & -M & -M & -M & -M & -M & \boxed{0} & -M & -M \\
			\vdots & -M & \ddots & -M & -M & -M & -M & -M & -M  & -M & \ddots & -M & -M & -M & -M & -M & -M & -M & \ddots & -M \\
			b_m & -M & -M & -x_{3m - 1} - x_{3m} & -M & -M & -M & -M & -M & -M & -M & \boxed{-x_{3m - 2}} & \boxed{-x_{3m - 1}} & \boxed{-x_{3m}} & -M & -M & -M & -M & -M & \boxed{0}
			\end{array}$}
	\end{center}
	where $\epsilon > 0$ is a constant small enough, $M$ is a constant greater than $m + 1$ and the $x_i$ are assumed to be ordered in a decreasing order: $\forall i \in \llbracket 1, 3m \rrbracket, x_i \geq x_{i + 1}$.
	
%	Utilities are such that: agent $a_i$, $i \in \llbracket 1, m \rrbracket$, gives value $-m - \epsilon$ to chore $g_i$ and $-M$ ot any other ``g" chore. Moreover, she values $-1-\epsilon$ chore $h_i$, 0 chore $h_{i - 1}$ (chore $h_m$ for agent $a_1$) and $-M$ the other ``h" chores. Her valuation for ``l" chores follows the opposite of the values in $X$. Finally her utility for ``o" chores is 0 for $o_i$ and $-M$ for the others.
%	
%	Agent $b_i, i \in \llbracket 1, m \rrbracket$, values $-x_{3i - 1} - x_{3_i}$ chores $g_i$, and $-M$ all other ``g" chores. All ``h" chores give her $-M$ utility. Chores $l_{3i - 2}, l_{3i - 1}$ and $l_{3i}$ respectively provide her $x_{3i - 2}, x_{3i - 1}$ and $x_{3i}$ utility while she values $-M$ any other ``l" chore.  Finally her utility for ``o" chores is 0 for $o_{i + m}$ and $-M$ for the others.
	
	The initial allocation $\pi \in \Pi(\objSet, \agentSet)$ is depicted by the boxed items in the previous table. It is defined as $\pi_{a_i} = \{g_i\} \cup \{h_i\} \cup \{o_i\}$ for every $i \in \llbracket 1, m \rrbracket$ and $\pi_{b_i} = \{l_{3i - 2}, l_{3i - 1}, l_{3i}\} \cup \{o_{m + i}\}$ for all $i \in \llbracket 1, m \rrbracket$.
	
	In the next step, omitted, we show that $\pi$ violates GEF1 if and only if there exists a suitable partition of $X$ satisfying the \textsc{3-Partition} problem. In particular, we show that if two groups $S$ and $T$ are witnesses of a violation of GEF1, then $S = T = \agentSet$. 
\end{proof}

In this proof, the only possible violations of GEF1 are such that $S = T = \agentSet$. Hence, checking whether an allocation satisfies the Pareto-optimality relaxation derived from GEF1, is also coNP-complete. These results are similar in flavour to the results of \citet{KBKZ09} and \citet{ABL+19a} that testing Pareto-optimality is coNP-complete.

\section{Conclusions}

% Inspired by the group envy-freeness concept and its relaxations for the case of divisible goods, we formalized group-envy freeness and several relaxations for indivisible goods and chores. This concept has both fairness and efficiency flavours. Our definitions are general and work well for ordinal preferences as well as cardinal utilities involving goods and chores. We clarified the relation of GEF1 with other fairness concepts and presented several positive and negative computational results.

Inspired by the group envy-freeness concept and its relaxations for the case of divisible goods, we formalized several relaxations of the concept for indivisible goods and chores. The concepts have both fairness and efficiency flavours. Our definitions are general and our key concepts work well for ordinal preferences as well as cardinal utilities involving goods and chores. We clarified the relation of GEF1 with other fairness concepts and presented several positive and negative computational results. 

Several interesting questions arise as a result of our study. The main question left open is the existence of GEF1 allocations when there are goods and chores. The question has been answered positively in the case of goods. However the proof involves the Nash social welfare which cannot be used with chores. Considering that protection of groups is one of the central concerns in new research on algorithmic fairness, we envisage GEF1 and its variants to spur further interesting work in the area.

%Answering the existence of GEF1 allocations requires then to develop new solution concepts behaving well with goods as well as with chores.
	
	\bibliographystyle{plainnat}
	\bibliography{abbshort,GEF}
	
	\appendix
	\chapter*{Appendix}
	
	\section{The Egal-Sequential Algorithm for Identical Utilities}
	
	\subsection{Proof of Lemma \ref{lemma:egalSeqAlgoCorrect}}
	We call $u$ the common utility function. We show by induction on the number of items that have already been allocated that the Egal-Sequential Algorithm maintains a partial allocation that is EFX. Let us denote by $\pi^k$ the partial allocation constructed after allocating the $k$-th item.
	
	The base case for $k = 1$ is straightforward. EFX is trivially satisfied by $\pi^1$ as only one item has been allocated.
	
	Let us now suppose that for a given $k < m$ the partial allocation $\pi^{k}$ is EFX. We show that $\pi^{k + 1}$ also satisfies EFX. Let $o$ be the item that is allocated to agent $i^*$ at the $k + 1$-th step of the algorithm. Let us distinguish two cases depending on whether $o$ is a good or a chore.
	\begin{itemize}
		\item If $u(o) \geq 0$, the only change in the utilities of the agents is an increase of $i^*$'s utility. The only possibility to violate EFX would therefore be for agents to become envious up to any item of $i^*$. As $i^* \in \arg \min_{i \in \agentSet} u(\pi^{k}_i)$, no agent envies $i^*$ in $\pi^{k}$: $\forall j \in \agentSet, u(\pi^{k}_{i^*}) \leq u(\pi^{k}_j)$. This implies that $\forall j \in \agentSet, u(\pi^{k + 1}_{i^*} \backslash \{o\}) \leq u(\pi^{k + 1}_j)$. Since $o$ is the smallest item allocated at step $k$, no agent envies $i^*$ up to any item in $\pi^{k + 1}$. Hence, $\pi^{k + 1}$ is EFX.
		\item If $u(o) < 0$, the only change in the utilities is a decrease of $i^*$'s utility. Therefore, the only possibility for EFX to be violated would be if $i^*$ becomes envious up to any item of another agent. Since $i^* \in \arg \max_{i \in \agentSet} u(\pi^{k}_i)$, $i^*$ does not envy anyone in $\pi^{k}$, that is, $\forall j \in \agentSet, u(\pi^{k}_{i^*}) \geq u(\pi^{k}_j)$. This implies that $\forall j \in \agentSet, u(\pi^{k + 1}_{i^*} \backslash \{o\}) \geq u(\pi^{k + 1}_j)$. Thus, $i^*$ is still not envious. Hence, $\pi^{k + 1}$ is still EFX.
	\end{itemize}

	We have then proved that the Egal-Sequential Algorithm returns EFX allocations when preferences are identical.
	
	Sorting items can be done in time in $\mathcal{O}(m \log m)$. The for loop of the algorithm uses $m$ steps during which finding an agent with maximum or minimum utility can be done in time in $\mathcal{O}(n)$ hence the overall complexity is $\mathcal{O}(\max\{m \log m, mn\})$.
		
	\subsection{Proof of Lemma \ref{lemma:EFX=>GEF1Identical}}
	
	Let us consider an allocation $\pi$ that satisfies EFX but not GEF1. Since preferences are identical we can assume, without loss of generality, that no item gives zero utility. If such items exist, they can be allocated to any agent without changing anything. As $\pi$ is not GEF1, there exist two groups $S \subseteq \agentSet$ and $T \subseteq \agentSet$ and a reallocation $\pi' \in \Pi(\pi_T, S)$ such that $\forall i \in S, \forall o \in \pi_i^- \cup \pi_i^+, u(\pi_i' \backslash \{o\}) \geq u(\pi_i \backslash \{o\})$ with one inequality being strict. For $i \in S$, we introduce $s_i$ defined as:
	$$s_i = \max\{\max_{o \in \pi_i'^+} u(o), \max_{o \in \pi_i^-} -u(o), 0\}.$$
	Violating GEF1 then implies: $\forall i \in S, u(\pi'_i) - u(\pi_i) \geq s_i$ with one inequality being strict. The 0 component in the definition of $s_i$ is meant to tackle the case when an agent $i$ receives an empty allocation in $\pi$ and no goods in $\pi'$. If $\pi_i = \pi_i'^+ = \emptyset$ we should have $\pi_i^- = \emptyset$ to get a GEF1 violation, hence we have $u(\pi_i) = u(\pi_i')$ and $s_i = 0$ is a suitable bound.
	
	By summing the inequalities coming from the violation of GEF1 over $i \in S$, we obtain $\sum_{i \in S} u(\pi_i') - \sum_{i \in S} u(\pi_i) > \sum_{i \in S} s_i$. Since utilities are identical and additive this implies:
	\begin{align}
	\label{eq:proofEFX=>GEF1Line1'}
	u(\pi_T) - u(\pi_S) > \sum_{i \in S} s_i.
	\end{align}
	
	Moreover, as $\pi$ is EFX, we have $\forall i, j \in \agentSet, \forall o \in \pi_i^- \cup \pi_j^+, u(\pi_i \backslash \{o\}) \geq u(\pi_j \backslash \{o\})$. For $i, j \in \agentSet$, we introduce $s_{i, j}$ defined by:
	$$s_{i, j} = -\min\{\min_{o \in \pi_j^+} -u(o), \min_{o \in \pi_i^-} u(o), 0\}.$$
	The 0 component is once again here to take care of the case when $\pi_i^- \cup \pi_j^+ = \emptyset$. We have thus:
	\begin{align}
	\label{eq:proofEFX=>GEF1Line2'}
	\forall i, j \in \agentSet, s_{i, j} \geq u(\pi_j) - u(\pi_i).
	\end{align}
	
	Our goal is now to sum up inequalities describe in \eqref{eq:proofEFX=>GEF1Line2'} to obtain a contradiction with \eqref{eq:proofEFX=>GEF1Line1'}. We are looking for a set of pairs $(i, j)$ such that each $i \in S$ and each $j \in T$ appear once and only once, and such that the sum of $s_{i, j}$ over these pairs is smaller than $\sum_{i \in S} s_i$. To do so, we find a suitable matching in a bipartite graph.
	
	\medskip
	
	Let us consider the bipartite graph $G = \tuple{S \cup T, E}$ where nodes represent agents in $S$ and in $T$. There is an edge $(i, j) \in E$ between agents $i \in S$ and $j \in T$ if and only if $\pi_i'^+ \cap \pi_j\neq \emptyset$, that is, $i$ receives some of $j$'s goods in $\pi'$. We consider a partition $S^+ \cup S^-$ of $S$ where:
	\begin{align*}
	S^+ &= \{ i \in S \mid u(\pi_i) \geq 0 \} \\
	S^- &= \{ i \in S \mid u(\pi_i) < 0 \}.
	\end{align*}
	For $X \subseteq S$, we write $N(X)$ its neighbourhood in the graph $G$: $N(X) = \{j \in T \mid \exists i \in X, (i, j) \in E\}$. A symmetric definition holds for $N(Y)$ where $Y \subseteq T$.
	
	\medskip
	
	We claim that there always exists a matching $M \subseteq S \times T$ in $G$ that matches all the agents in $S^+$. Suppose for the sake of contradiction that such $M$ does not exist. From Hall's theorem \citep{Hall35}, there must exist $X\subseteq S^+$ and $Y \subseteq T$ such that $Y = N(X)$ and $|Y| < |X|$. Let us assume that $X$ is a smallest Hall's violation in $G$, and consider $M' \subseteq X \times Y$ a maximum matching between agents in $X$ and in $Y$. 
	
	We show in the following that $M'$ always match all the agents in $Y$. To get a better understanding, a diagram illustrating the different set of agents considered is presented in Figure \ref{fig:proofEFX=>GEF1}. Assume that $M'$ does no match all the agents in $Y$, Hall's theorem implies that there exist $X' \subseteq X$ and $Y' \subseteq Y$ such that $X' = N(Y')$ and $|X'| < |Y'|$. Then for $X'' = X \backslash X'$ and $Y'' = Y \backslash Y'$, we have $|X''| > |Y''|$ and $N(X'') = Y''$. Indeed, as $|X| > |Y|$ we have $|X| - |X'| > |Y| - |X'|$ and thus $|X''| > |Y''|$ since $|X'| < |Y'|$. Moreover it is clear that $N(X'') = Y''$ as otherwise any agent in $X''$ linked to another agent in $Y'$ is in $N(Y') = X'$ and can not be in $X''$. Overall, $X''$ constitutes a Hall's violation for the existence of $M$ which is smaller than $X$. This contradicts the fact that $X$ is the smallest such Hall's violation. Sets $X'$ and $Y'$ does not exist and $M'$ matches all the agents in $Y$.
	\begin{figure}
		\centering
		\resizebox{!}{4cm}{
			\begin{tikzpicture}
			\node[] at (1, 4.4) {$S$};
			
			\draw[] (0, 3.5) rectangle (2, 4);
			\node[] at (1, 3.75) {$S^-$};
			\draw[] (0, 0) rectangle (2, 3.5);
			\node[] at (1, 3.125) {$S^+$};
			
			%	 		\node[] at (1, 3.125) {$X$};
			\draw[] (0.25, 2) rectangle (1.75, 2.75);
			\node[] at (1, 2.375) {$X'$};
			\draw[] (0.25, 0.25) rectangle (1.75, 2);
			\node[] at (1, 1.125) {$X''$};
			
			\draw[] (6, 0) rectangle (8, 4);
			\node[] at (7, 4.4) {$T$};
			
			% 			\node[] at (7, 6.5) {$Y$};
			\draw[] (6.25, 1.25) rectangle (7.75, 2.5);
			\node[] at (7, 1.875) {$Y'$};
			\draw[] (6.25, 0.5) rectangle (7.75, 1.25);
			\node[] at (7, 0.875) {$Y''$};
			\draw[dashed] (1.75, 2.75) -- (6.25, 2.5);
			\draw[dashed] (1.75, 2) -- (6.25, 1.25);
			\draw[dashed] (1.75, 0.25) -- (6.25, 0.5);
			\end{tikzpicture}}
		\caption{Sets of agents considered in the proof of Lemma \ref{lemma:EFX=>GEF1Identical}.}
		\label{fig:proofEFX=>GEF1}
	\end{figure}
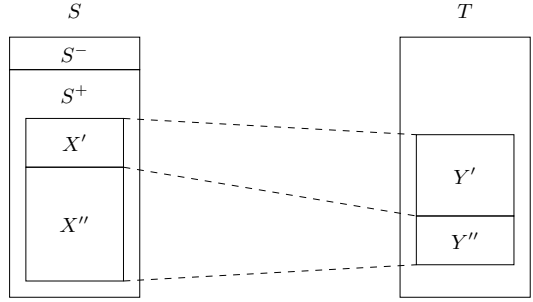
	
	We now turn back to the problem of showing that $M$ matches all the agents in $S^+$. As $M'$ matches all the agents in $Y$ and $u$ is additive, we have $\sum_{(i, j) \in M'} u(\pi_j) = u(\pi_Y)$. Hence, by summing inequalities \eqref{eq:proofEFX=>GEF1Line2'} for $(i, j) \in M'$, we obtain:
	\begin{align}
	\label{eq:proofEFX=>GEF1Line3'}
	u(\pi_Y) - \sum_{(i, j) \in M'} s_{i, j} \leq \sum_{(i, j) \in M'} u(\pi_i)
	\end{align}
	Moreover summing inequalities yielded by the violation of GEF1 over $i \in S$ in a similar manner as inequalities \eqref{eq:proofEFX=>GEF1Line1'} brings $u(\pi_Y) - \sum_{i \in X} s_i \geq u(\pi_X)$. Together with \eqref{eq:proofEFX=>GEF1Line3'}, this leads to:
	\begin{align}
	\label{eq:proofEFX=>GEF1Line4'}
	\sum_{(i, j) \in M'} u(\pi_i) + \sum_{(i, j) \in M'} s_{i, j} \geq u(\pi_Y) \geq u(\pi_X) + \sum_{i \in X} s_i.
	\end{align}
	Observe that for any pair of agents $(i, j)$ such that there is an edge between $i$ and $j$ in $G$, we have:
	\begin{align}
	\label{eq:proofEFX=>GEF1Line5'}
	s_i \geq \max_{o \in \pi_i'^+} u(o) \geq \max_{o \in \pi_i'^+ \cap \pi_j} u(o) \geq \min_{o \in \pi_j^+} u(o) \geq s_{i, j}.
	\end{align}
	In addition, as $|M'| < |X|$, it is clear that $\sum_{(i, j) \in M'} u(\pi_i) < \sum_{i \in X} u(\pi_i)$ and thus $\sum_{(i, j) \in M'} u(\pi_i) < u(\pi_X)$. Overall, we have:
	$$\sum_{(i, j) \in M'} u(\pi_i) + \sum_{(i, j) \in M'} s_{i, j} < u(\pi_X) + \sum_{i \in X} s_i,$$
	which contradicts \eqref{eq:proofEFX=>GEF1Line4'}. Therefore, no Hall's violation can exist and we have proved that the matching $M$ does match all the agents in $S^+$. 
	
	\medskip
	
	Next we show that the existence of this matching $M$ leads to a contradiction on the fact that $\pi$ is not GEF1. We extend the matching $M$ to match all the agents in $S$ by arbitrarily pairing each agent $i\in S^-$ with an unmatched agent in $T$. Let the extended matching be called $M^*$. Observe that for any new pair of agents $(i, k) \in M^* \backslash M$, we have $i \in S^-$, that is $\pi_i^- \neq \emptyset$. Hence, for any agent $j \in T$, we have:
	\begin{align}
	\label{eq:proofEFX=>GEF1Line6'}
	s_{i, j} \leq \min_{o \in \pi_i^-} u(o) \leq \max_{o \in \pi_i^-} u(o) \leq s_i.
	\end{align}
	By summing \eqref{eq:proofEFX=>GEF1Line2'} over $(i, j) \in M^*$ we obtain $\sum_{(i, j) \in M^*} s_{i, j} \geq u(\pi_T) - u(\pi_S)$. From \eqref{eq:proofEFX=>GEF1Line5'} and \eqref{eq:proofEFX=>GEF1Line6'} we get that $\sum_{(i, j) \in M^*} s_{i, j} \leq \sum_{i \in S} s_i$, hence summing \eqref{eq:proofEFX=>GEF1Line2'} over $(i, j) \in M^*$ leads to a contradiction with \eqref{eq:proofEFX=>GEF1Line1'}. We have thus proved that $\pi$ satisfies both EFX and GEF1.
	
	\subsection{Proof of Theorem \ref{thm:existGEF1Identical}}
	
	The proof is trivial, it just amounts at linking Lemma \ref{lemma:egalSeqAlgoCorrect} and Lemma \ref{lemma:EFX=>GEF1Identical}.
	
	\section{The Ternary Flow Algorithm}
	
	Let us first introduce briefly the theory of network flows as well as our notations. A cost flow network is represented by a directed graph $G = \tuple{V, E}$ whose nodes $v \in V$ are labelled with a demand $d(v) \in \mathbb{R}$ while edges $e \in E$ are labelled by a maximum capacity $\delta(e) \in \mathbb{R}_{> 0}$ and a cost $c(e) \in \mathbb{R}_{\geq 0}$. A flow $f$ is a mapping $f: E \rightarrow \mathbb{R}_{\geq 0}$ where $f(e)$ is the amount of flow passing through the edge $e$. By a slight abuse of notation, we use $f(v)$ for a node $v \in V$ to denote the difference between the outgoing flow and the incoming flow in $v$. A flow is realizable if for every edge $e \in E$ its capacity is not exceeded, $f(e) \leq \delta(e)$, and for every node $v \in V$ its demand is satiated, $f(v) = d(v)$. The cost of a realizable flow $f$ is defined by $c(f) = \sum_{e \in E} f(e) c(e)$. A minimum cost flow is then a flow with minimal cost among the realizable flows. A flow is said to be integer if $\forall e \in E, f(e) \in \mathbb{N}$. 
	
	Then we present the Nash Flow Algorithm proposed by \cite{DaSc15}. Let $I = \tuple{\agentSet, \objSet, (u_i)_{i\in\agentSet}} \in \instGood$ be an instance with only goods and where preferences are normalized ternary symmetric. Note that since there are only goods, the utilities for the singletons are in $\{0, 1\}$. The network $G = \tuple{V, E}$ consists in the set of nodes $V = \{s, t\} \cup \agentSet \cup \objSet \cup \{t_{i, j} \mid i \in \agentSet, o_j \in \objSet\}$ where the demands are $d(s) = m$, $d(t) = -m$ and $d(v) = 0$ for every $v \in V \backslash \{s, t\}$. The edge set $E$ is defined by:
	\begin{itemize}
		\item for every item $o \in \objSet$, the edge $(s, o)$ is in $E$ with capacity $\delta(s, o) = 1$ and cost $c(s, o) = 0$,
		\item for every agent $i \in \agentSet$, for every item $o \in \objSet$ such that $u_i(o) = 1$, the edge $(o, i)$ is in $E$ with capacity $\delta(o, i) = 1$ and cost $c(o, i) = 0$,
		\item for every agent $i \in \agentSet$ and every item $o_j \in \objSet$, two edges are in $E$: $(i, t_{i, j})$ with capacity $\delta(i, t_{i, j}) = 1$ and cost $c(i, t_{i, j}) = n^j$ and $(t_{i, j}, t)$ with capacity $\delta(t_{i, j}, t) = 1$ and cost $c(t_{i, j}, t) = 0$. In the following we refer to $(i, t_{i, j})$ edges as ``t'' edges.
	\end{itemize}
	Since the capacities are all integer, there always exists a minimum cost integer flow of $G$ by the integrality property. We have moreover a one-to-one correspondence between integer flows in $G$ and allocations in $\Pi(\agentSet, \objSet)$. From a flow $f$, an allocation $\pi$ is defined by $o \in \pi_i$ if and only if $f(o, i) = 1$. From an allocation $\pi$ the flow $f$ is defined by $f(s, o) = 1$ for each $o \in \objSet$, $f(o, i) = 1$ if and only if $o \in \pi_i$ for every $o \in \objSet$ and $i \in \agentSet$, and $f(i, t_{i, j}) = f(t_{i, j}, t) = 1$ for every $i \in \agentSet$ and $1 \leq j \leq u_i(\pi_i)$. An allocation $\pi$ maximizes the Nash social welfare if and only if it corresponds to a minimum cost integer flow of the network $G$.
		
	\subsection{Proof of Lemma \ref{lemma:TernNashFlow=>Lex}}
	
	For a flow $f$, we denote by $f_i$, $i \in \agentSet$, the amount of flow passing through agent node $i$, that is $f_i = \sum_{j \in \objSet} f(j, i)$. Note that $f_i$ is then the number of good received by agent $i$ in the allocation corresponding to $f$. Since utilities are normalized ternary, it is also $i$'s utility, $u_i(\pi_i)$.
	
	We first show that a leximin-optimal allocation $\pi$ corresponds to a minimum cost flow $f$. Let us suppose toward a contradiction that there exists two agents $i_1$ and $i_2$ with $f_{i_1} - f_{i_2} \geq 2$ and a good $o$ such that $(o, i_2) \in E$ and $f(o, i_2) = 1$. We have thus $u_{i_1} (o) = u_{i_2} (o) = 1$. Consider then the allocation $\pi'$ such that $\pi'_k = \pi_k, \forall k \in \agentSet\backslash\{i_1, i_2\}$, $\pi'_{i_1} = \pi_{i_1} \backslash \{o\}$ and $\pi'_{i_2} = \pi_{i_2} \cup \{o\}$ and denote by $f'$ the corresponding flow. Since $f_{i_1} - f_{i_2} \geq 2$, agent $i_2$ comes before agent $i_1$ in the lexmin ordering of the utilities in $\pi$. It is also the case in $\pi'$ since $f'_{i_1} - f'_{i_2} = f_{i_1} - f_{i_2} - 2$. Moreover we have $u_{i_2}(\pi'_{i_2}) > u_{i_2}(\pi_{i_2})$, hence $\vec{u}(\pi') \succ_{lex} \vec{u}(\pi)$ which is a contradiction. Agents $i_1$ and $i_2$ do not exist which is a sufficient condition\footnote{\citet[proof of Theorem 3.3, step 1]{DaSc15} showed that an integer flow $f$ is a minimum cost flow if and only if $\forall i \in \agentSet, \forall h\text{ s.t. }1 \leq h \leq f_i, f(i, t_{i, h}) = 1$ and there is no sequence $\tuple{i_1, o_1, i_2, o_2, \ldots, o_{l - 1}, i_l}$ with $f_{i_1} - f_{i_l}$ such that for all $1 \leq h \leq l - 1$, we have $(o_h, i_{h + 1}) \in E$ and $f(o_h, i_h) = 1$.} for $f$ to be a minimum cost flow. 
	
	Next we show the other implication. For a flow $f$, we write $\vec{f}$ the vector of the $f_i$ ordered from the lowest to the highest. Observe that $\vec{f} = \vec{u}(\pi)$ for $\pi$ the allocation corresponding to $f$. \citet[proof of Theorem 3.3, step 3]{DaSc15} proved that all minimum cost integer flows have the same vector $\vec{f}$. We just proved that there exists a minimum cost integer flow $f^*$ corresponding to a leximin-optimal allocation $\pi^*$, hence any minimum cost integer flow $f$ is such that $\vec{f} = \vec{f^*}$, which proves that the allocation $\pi$ corresponding to $f$ is leximin-optimal.
	
	\subsection{Proof of Lemma \ref{lemma:algoLeminTernCorrect}}
	
	We show that along the execution of Algorithm \ref{algo:lexminTern}, the partial allocation $\pi$ is always leximin-optimal. 
	
	Based on Lemma \ref{lemma:TernNashFlow=>Lex}, we know that after line \ref{algoLine:lexminTernNashFlow} of Algorithm \ref{algo:lexminTern}, $\pi$ is leximin-optimal for the utility profile $(u_i')_{i \in \agentSet}$. Since agents are only allocated items of utility 1, $\pi$ is also leximin-optimal for $(u_i^{\text{Norm}})_{i \in \agentSet}$.
	
	Next, we show by induction that $\pi$ is always leximin-optimal according to the normalized utilities throughout the first ``for loop'' (line \ref{algoLine:lexminTernFor1}) where items in $O^-$ are allocated. In the following, leximin-optimality is according to the normalized utilities. Consider one step of the loop, we call $o \in O^-$ the chore that is to be allocated, $\mu$ the current allocation, that is leximin-optimal, and $\pi$ the allocation where $o$ has been allocated to agent $i^* \in \arg\max_{i \in \agentSet} u_i^{\text{Norm}}(\pi_i)$. In the following, we use $\vec{u}$ to refer to the leximin ordering according to preferences $(u_i^{\text{Norm}})_{i \in \agentSet}$. Suppose that $\pi$ is not leximin-optimal, then there exist an allocation $\pi'$ so that there exists an integer $k \in \llbracket 1, n \rrbracket$ such that $\vec{u}_j(\pi') = \vec{u}_j(\pi), \forall j < k$, and $\vec{u}_k(\pi') > \vec{u}_k(\pi)$. Since $i^* \in \arg\max_{i \in \agentSet} u_i^{\text{Norm}}(\pi_i)$ we can assume without loss of generality that $i^*$'s index in $\vec{u}(\pi)$ is $n$. From the induction hypothesis we know that $\mu$ is leximin-optimal, hence $n \leq k$ as otherwise the allocation obtained from $\pi'$ by deleting chore $o$ would leximin-dominates $\mu$. We have thus $k = n$, that is $\sum_{i \in \agentSet} \vec{u}_i(\pi') > \sum_{i \in \agentSet} \vec{u}_i(\pi)$. However since $o$ is a chore for every agent and $\mu$ is leximin-optimal, we know from Lemma \ref{prop:characPOTernaryUtil} that $\pi$ is Pareto-optimal. Moreover, $\pi'$is also Pareto-optimal, since it is leximin-optimal, Lemma \ref{prop:characPOTernaryUtil} implies then $\sum_{i \in \agentSet} \vec{u}_i(\pi) = \sum_{i \in \agentSet} \vec{u}_i(\pi')$. This is a contraction. We have then proved that the allocation $\pi$ is still leximin-optimal after allocating chore $o$. The induction is settled and after the first ``for loop'', the partial allocation $\pi$ is leximin-optimal.
	
	Finally, note that the items in $O^0$ are allocated to agents that have 0 utility for them. Hence, the allocation $\pi$ is still leximin-optimal after the second ``for loop''. After this last loop, $\pi$ is no longer a partial allocation, all the items have been allocated, hence Algorithm \ref{algo:lexminTern} returns $\pi$ that is leximin-optimal according to the normalized utilities.
	
	\subsection{Proof of Lemma \ref{lemma:lexmin=>GEF1Tern}}
	
	In the following we consider an allocation $\pi$, supposed to be leximin-optimal. We first show that the agents' utility can not be too different because of leximin-optimality.
	\begin{claim}
		For every $i, j \in \agentSet$, if $u_j(\pi_j) - u_i(\pi_i) \geq 2$ then $\{o \in \pi_j \mid u_i(o) = 1\} \cup \{o \in \pi_i \mid u_j(o) = -1\} = \emptyset$.
	\end{claim}
	\begin{claimproof}
		Assume that there is an item $o \in \pi_j$ such that $u_i(o) = 1$. Consider the allocation $\pi'$ such that $\forall k \in \agentSet \backslash \{i, j\}, \pi'_k = \pi_k$, $\pi'_i = \pi_i \cup \{o\}$, $\pi'_j = \pi_j \backslash \{o\}$. Since $u_j(\pi_j) - u_i(\pi_i) \geq 2$, $i$ appears before $j$ in $\vec{u}(\pi)$. Moreover, as $u_j(\pi_j') - u_i(\pi_i') \geq 0$, $i$ also appears before $j$ in $\vec{u}(\pi')$. However, $u_i(\pi'_i) > u_i(\pi_i)$ implies that $\vec{u}(\pi') \succ_{lex} \vec{u}(\pi)$ which contradicts the fact that $\pi$ is leximin-optimal. The case where there exists an item $o \in \pi_i$ such that $u_j(o) = -1$ is exactly symmetric. 
	\end{claimproof}
	
	Next, we show our second claim that if one agent has negative utility for her bundle, then all agent's utility is at most one more.
	\begin{claim}
		If $min_{i \in \agentSet}u_i(\pi_i) < 0$, then $\forall i \in \agentSet, 0 \leq u_i(\pi_i) - min_{i \in \agentSet}u_i(\pi_i) \leq 1$. 
	\end{claim}
	\begin{claimproof}
		Let $k \in \agentSet$ be the agent with minimum utility in $\pi$. Assume that $u_k(\pi_k) < 0$, it means that $k$ owns at least one chore $o$ in $\pi$. Since $\pi$ is Pareto-efficient, from Lemma \ref{prop:characPOTernaryUtil} we know that $o$ is a chore for every agent in $\agentSet$. Let us assume that there exists an agent $j \in \agentSet$ such that $u_j(\pi_j) > u_k(\pi_k) + 1$, then since $o$ is also a chore for $j$, the allocation obtained by transferring $o$ from $k$ to $j$ would leximin-dominates $\pi$. This would contradict the leximin-optimality of $\pi$. 
	\end{claimproof}
	
	For the sake of contradiction, assume that $\pi$ is not GEF1, that is, there exist two groups of agents $S, T \subseteq \agentSet$, and a reallocation $\pi' \in \Pi(\pi_T, S)$ such that $\forall i \in S, \forall o \in \pi_i^- \cup \pi_i^+, u(\pi_i' \backslash \{o\}) \geq u(\pi_i \backslash \{o\})$ with one inequality being strict. Suppose without loss of generality that the utilities of the agents in $S$ (resp. $T$), written $s_1, \ldots, s_{|S|}$ (resp. $t_1, \ldots, t_{|S|}$), are ordered increasingly. From Lemma \ref{prop:characPOTernaryUtil} we know that for every agent $j \in T$ and $i \in S$, $u_i(\pi_j) \leq u_j(\pi_j)$. Hence, $t_j$ is an upper bound on the utility $i$ can receive from $j$. 
	
	To get a GEF1 violation, a reallocation of the items in $\pi_T$ should give utility at least $s_i + 1$ to every agent $i \in S$ with one agent receiving strictly more. Let us denote by $j^* \in T$ the index of the first $t_j$ such that $t_j > s_j + 1$. From the second Claim, it cannot be the case that $s_{j^*} <0 $. Let us then assume that $s_{j^*} \geq 0$. From our first Claim, we know that in this case $j^*$ does not have any item considered as goods for the agent corresponding to $s_{j^*}$, written $i^*$. Since utilities are ordered increasingly, it is also the case for every agent $j > j^*$. Hence $i^*$ can only receive goods from agents $j < j^*$. However, for every agent $j < j^*$, we have $s_j = t_j + 1$. These agents can thus not provide enough goods to all agents $i \leq i^*$ to get GEF1 violation. This contradicts the existence of such agent $j^*$, which in turns contradicts the existence of a GEF1 violation. We have thus proved that $\pi$ is both leximin-optimal and GEF1.
	
	\subsection{Example of the Ternary Flow Algorithm}
	
	Let us illustrate on an example how Algorithm \ref{algo:lexminTern} proceeds. Consider once again the instance presented in Example \ref{ex:PO + EF1 =/> GEF1}. The network flow constructed by the Nash Flow Algorithm is presented in Figure \ref{fig:exFlowAlgo}. The allocation corresponding to the minimum cost integer flow is $\pi = \tuple{\{o_1, o_2\}, \{o_3, o_4\}, \{o_5, o_6\}, \{o_7, o_8\}}$. Since all items are considered as good by at least one agent, the allocation returned by the Nash Flow Algorithm is also the one returned by the Ternary Flow Algorithm.
	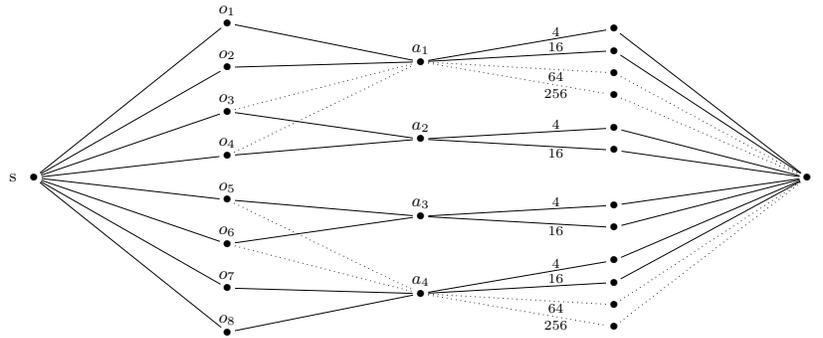
\begin{figure}
		\centering
		\resizebox{!}{14em}{
			\begin{tikzpicture}[shorten < = 4 pt, shorten > = 4pt]
			\def\hDist{3.5}
			\def\vSepO{0.8}
			\def\vSepN{1.4}
			\def\vSepT{0.4}
			
			\node[label = left:s] at (0, 0) {$\bullet$};
			\node[label = right:t] at ($(4 * \hDist, 0)$) {$\bullet$};

			\node[label = {[label distance=-0.5em]90:$o_1$}] at ($(\hDist, \vSepO / 2 + 3 * \vSepO)$) {$\bullet$};
			\draw[black] (0, 0) -- ($(\hDist, \vSepO / 2 + 3 * \vSepO)$);
			
			\node[label = {[label distance=-0.5em]90:$o_2$}] at ($(\hDist, \vSepO / 2 + 2 * \vSepO)$) {$\bullet$};
			\draw[black] (0, 0) -- ($(\hDist, \vSepO / 2 + 2 * \vSepO)$);
			
			\node[label = {[label distance=-0.5em]90:$o_3$}] at ($(\hDist, \vSepO / 2 + \vSepO)$) {$\bullet$};
			\draw[black] (0, 0) -- ($(\hDist, \vSepO / 2 + \vSepO)$);
			
			\node[label = {[label distance=-0.5em]90:$o_4$}] at ($(\hDist, \vSepO / 2)$) {$\bullet$};
			\draw[black] (0, 0) -- ($(\hDist, \vSepO / 2)$);
			
			\node[label = {[label distance=-0.5em]90:$o_5$}] at ($(\hDist, - \vSepO / 2)$) {$\bullet$};
			\draw[black] (0, 0) -- ($(\hDist, -\vSepO / 2)$);
			
			\node[label = {[label distance=-0.5em]90:$o_6$}] at ($(\hDist, - \vSepO / 2 - \vSepO)$) {$\bullet$};
			\draw[black] (0, 0) -- ($(\hDist, -\vSepO / 2 - \vSepO)$);
			
			\node[label = {[label distance=-0.5em]90:$o_7$}] at ($(\hDist, - \vSepO / 2 - 2 * \vSepO)$) {$\bullet$};
			\draw[black] (0, 0) -- ($(\hDist, -\vSepO / 2 - 2 * \vSepO)$);
			
			\node[label = {[label distance=-0.5em]90:$o_8$}] at ($(\hDist, - \vSepO / 2 - 3 * \vSepO)$) {$\bullet$};
			\draw[black] (0, 0) -- ($(\hDist, -\vSepO / 2 - 3 * \vSepO)$);

			\node[label = {[label distance=-0.5em]90:$a_1$}] at ($(2 * \hDist, \vSepN / 2 + \vSepN)$) {$\bullet$};
			\draw[black] ($(\hDist, \vSepO / 2 + 3 * \vSepO)$) -- ($(2 * \hDist, \vSepN / 2 + \vSepN)$);
			\draw[black] ($(\hDist, \vSepO / 2 + 2 * \vSepO)$) -- ($(2 * \hDist, \vSepN / 2 + \vSepN)$);
			\draw[dotted] ($(\hDist, \vSepO / 2 + \vSepO)$) -- ($(2 * \hDist, \vSepN / 2 + \vSepN)$);
			\draw[dotted] ($(\hDist, \vSepO / 2)$) -- ($(2 * \hDist, \vSepN / 2 + \vSepN)$);
			\node[] at ($(3 * \hDist, \vSepN / 2 + \vSepN + \vSepT / 2 + \vSepT)$) {$\bullet$};
			\node[] at ($(3 * \hDist, \vSepN / 2 + \vSepN + \vSepT / 2)$) {$\bullet$};
			\node[] at ($(3 * \hDist, \vSepN / 2 + \vSepN + -\vSepT / 2)$) {$\bullet$};
			\node[] at ($(3 * \hDist, \vSepN / 2 + \vSepN + -\vSepT / 2 - \vSepT)$) {$\bullet$};
			\draw[black] ($(2 * \hDist, \vSepN / 2 + \vSepN)$) -- ($(3 * \hDist, \vSepN / 2 + \vSepN + \vSepT / 2 + \vSepT)$);
			\draw[black] ($(2 * \hDist, \vSepN / 2 + \vSepN)$) -- ($(3 * \hDist, \vSepN / 2 + \vSepN + \vSepT / 2)$);
			\draw[dotted] ($(2 * \hDist, \vSepN / 2 + \vSepN)$) -- ($(3 * \hDist, \vSepN / 2 + \vSepN + -\vSepT / 2)$);
			\draw[dotted] ($(2 * \hDist, \vSepN / 2 + \vSepN)$) -- ($(3 * \hDist, \vSepN / 2 + \vSepN + -\vSepT / 2 - \vSepT)$);
			\draw[black] ($(3 * \hDist, \vSepN / 2 + \vSepN + \vSepT / 2 + \vSepT)$) -- ($(4 * \hDist, 0)$);
			\draw[black] ($(3 * \hDist, \vSepN / 2 + \vSepN + \vSepT / 2)$) -- ($(4 * \hDist, 0)$);
			\draw[dotted] ($(3 * \hDist, \vSepN / 2 + \vSepN + -\vSepT / 2)$) -- ($(4 * \hDist, 0)$);
			\draw[dotted] ($(3 * \hDist, \vSepN / 2 + \vSepN + -\vSepT / 2 - \vSepT)$) -- ($(4 * \hDist, 0)$);
			\node[] at ($(2.7 * \hDist, \vSepN / 2 + \vSepN + \vSepT / 2 + \vSepT - 0.05)$) {\scriptsize{$4$}};
			\node[] at ($(2.7 * \hDist, \vSepN / 2 + \vSepN + 4/5 * \vSepT - 0.05)$) {\scriptsize{$16$}};
			\node[] at ($(2.7 * \hDist, \vSepN / 2 + \vSepN - 8/5 * \vSepT / 2 + 0.05)$) {\scriptsize{$64$}};
			\node[] at ($(2.7 * \hDist, \vSepN / 2 + \vSepN - \vSepT / 2 - \vSepT + 0.02)$) {\scriptsize{$256$}};

			\node[label = {[label distance=-0.5em]90:$a_2$}] at ($(2 * \hDist, \vSepN / 2)$) {$\bullet$};
			\draw[black] ($(\hDist, \vSepO / 2 + \vSepO)$) -- ($(2 * \hDist, \vSepN / 2)$);
			\draw[black] ($(\hDist, \vSepO / 2)$) -- ($(2 * \hDist, \vSepN / 2 )$);
			\node[] at ($(3 * \hDist, \vSepN / 2 + \vSepT / 2)$) {$\bullet$};
			\node[] at ($(3 * \hDist, \vSepN / 2 - \vSepT / 2)$) {$\bullet$};
			\draw[black] ($(2 * \hDist, \vSepN / 2)$) -- ($(3 * \hDist, \vSepN / 2 + \vSepT / 2)$);
			\draw[black] ($(2 * \hDist, \vSepN / 2)$) -- ($(3 * \hDist, \vSepN / 2 - \vSepT / 2)$);
			\draw[black] ($(3 * \hDist, \vSepN / 2 + \vSepT / 2)$) -- ($(4 * \hDist, 0)$);
			\draw[black] ($(3 * \hDist, \vSepN / 2 - \vSepT / 2)$) -- ($(4 * \hDist, 0)$);
			\node[] at ($(2.7 * \hDist, \vSepN / 2 + \vSepT / 2 + 0.06)$) {\scriptsize{$4$}};
			\node[] at ($(2.7 * \hDist, \vSepN / 2 - \vSepT / 2 - 0.06)$) {\scriptsize{$16$}};

			\node[label = {[label distance=-0.5em]90:$a_3$}] at ($(2 * \hDist, - \vSepN / 2)$) {$\bullet$};
			\draw[black] ($(\hDist, -\vSepO / 2 - \vSepO)$) -- ($(2 * \hDist, -\vSepN / 2 )$);
			\draw[black] ($(\hDist, -\vSepO / 2)$) -- ($(2 * \hDist, -\vSepN / 2 )$);
			\node[] at ($(3 * \hDist, -\vSepN / 2 + \vSepT / 2)$) {$\bullet$};
			\node[] at ($(3 * \hDist, -\vSepN / 2 - \vSepT / 2)$) {$\bullet$};
			\draw[black] ($(2 * \hDist, -\vSepN / 2)$) -- ($(3 * \hDist, -\vSepN / 2 + \vSepT / 2)$);
			\draw[black] ($(2 * \hDist, -\vSepN / 2)$) -- ($(3 * \hDist, -\vSepN / 2 - \vSepT / 2)$);
			\draw[black] ($(3 * \hDist, -\vSepN / 2 + \vSepT / 2)$) -- ($(4 * \hDist, 0)$);
			\draw[black] ($(3 * \hDist, -\vSepN / 2 - \vSepT / 2)$) -- ($(4 * \hDist, 0)$);
			\node[] at ($(2.7 * \hDist, -\vSepN / 2 + \vSepT / 2 + 0.06)$) {\scriptsize{$4$}};
			\node[] at ($(2.7 * \hDist, -\vSepN / 2 - \vSepT / 2 - 0.06)$) {\scriptsize{$16$}};

			\node[label = {[label distance=-0.5em]90:$a_4$}] at ($(2 * \hDist, -\vSepN / 2 - \vSepN)$) {$\bullet$};
			\draw[black] ($(\hDist, -\vSepO / 2 - 3 * \vSepO)$) -- ($(2 * \hDist, -\vSepN / 2 - \vSepN)$);
			\draw[black] ($(\hDist, -\vSepO / 2 - 2 * \vSepO)$) -- ($(2 * \hDist, -\vSepN / 2 - \vSepN)$);
			\draw[dotted] ($(\hDist, -\vSepO / 2 - \vSepO)$) -- ($(2 * \hDist, -\vSepN / 2 - \vSepN)$);
			\draw[dotted] ($(\hDist, -\vSepO / 2)$) -- ($(2 * \hDist, -\vSepN / 2 - \vSepN)$);
			\node[] at ($(3 * \hDist, -\vSepN / 2 - \vSepN + \vSepT / 2 + \vSepT)$) {$\bullet$};
			\node[] at ($(3 * \hDist, -\vSepN / 2 - \vSepN + \vSepT / 2)$) {$\bullet$};
			\node[] at ($(3 * \hDist, -\vSepN / 2 - \vSepN - \vSepT / 2)$) {$\bullet$};
			\node[] at ($(3 * \hDist, -\vSepN / 2 - \vSepN - \vSepT / 2 - \vSepT)$) {$\bullet$};
			\draw[black] ($(2 * \hDist, -\vSepN / 2 - \vSepN)$) -- ($(3 * \hDist, -\vSepN / 2 - \vSepN + \vSepT / 2 + \vSepT)$);
			\draw[black] ($(2 * \hDist, -\vSepN / 2 - \vSepN)$) -- ($(3 * \hDist, -\vSepN / 2 - \vSepN + \vSepT / 2)$);
			\draw[dotted] ($(2 * \hDist, -\vSepN / 2 - \vSepN)$) -- ($(3 * \hDist, -\vSepN / 2 - \vSepN + -\vSepT / 2)$);
			\draw[dotted] ($(2 * \hDist, -\vSepN / 2 - \vSepN)$) -- ($(3 * \hDist, -\vSepN / 2 - \vSepN + -\vSepT / 2 - \vSepT)$);
			\draw[black] ($(3 * \hDist, -\vSepN / 2 - \vSepN + \vSepT / 2 + \vSepT)$) -- ($(4 * \hDist, 0)$);
			\draw[black] ($(3 * \hDist, -\vSepN / 2 - \vSepN + \vSepT / 2)$) -- ($(4 * \hDist, 0)$);
			\draw[dotted] ($(3 * \hDist, -\vSepN / 2 - \vSepN + -\vSepT / 2)$) -- ($(4 * \hDist, 0)$);
			\draw[dotted] ($(3 * \hDist, -\vSepN / 2 - \vSepN + -\vSepT / 2 - \vSepT)$) -- ($(4 * \hDist, 0)$);
			\node[] at ($(2.7 * \hDist, -\vSepN / 2 - \vSepN + \vSepT / 2 + \vSepT - 0.05)$) {\scriptsize{$4$}};
			\node[] at ($(2.7 * \hDist, -\vSepN / 2 - \vSepN + 4/5 * \vSepT - 0.05)$) {\scriptsize{$16$}};
			\node[] at ($(2.7 * \hDist, -\vSepN / 2 - \vSepN - 8/5 * \vSepT / 2 + 0.05)$) {\scriptsize{$64$}};
			\node[] at ($(2.7 * \hDist, -\vSepN / 2 - \vSepN - \vSepT / 2 - \vSepT + 0.02)$) {\scriptsize{$256$}};
			\end{tikzpicture}}
		\caption{Flow network constructed by the Nash Flow Algorithm on the instance described in Example \ref{ex:PO + EF1 =/> GEF1}. Plain lines represent a minimum cost flow and dotted lines any other edge. The cost of the edges are indicated next to them, 0 costs have been omitted.}
		\label{fig:exFlowAlgo}
	\end{figure}

	\section{Testing GEF1 is coNP-complete}
	
	All the reductions are from the \textsc{3-partition} problem, shown to be strongly NP-complete by \citet{GaJo75}. The problem is stated as follow:
	\begin{center}
		\begin{tabular}{rl}
			\toprule
			\multicolumn{2}{c}{\textsc{3-Partition}}\\
			\midrule
			\textbf{Instance:} & A multi-set of $3m$ numbers $X = \{x_1, \ldots, x_{3m}\}$ such that:\\
			& $\forall x \in X, 1/4 < x < 1/2$ and $\sum_{x \in X} x = m$.\\
			\textbf{Question:} & Is there a partition $(X_i)_{i \in \llbracket 1, m \rrbracket}$ of $X$ s.t. $\forall i, \sum_{x \in X_i} = 1$ ? \\
			\bottomrule
		\end{tabular}
	\end{center}
	
	\subsection{Proof of Theorem \ref{thm:complexityTestingGEF1} for \textsc{is-\GFoGood{}}}
	
	Let us consider $X = \{x_1, \ldots, x_{3m}\}$ an instance of the \textsc{3-Partition} problem. We construct an instance of the \textsc{is-\GFoGood{}} problem as follows. The set of agents is $\agentSet = \{a_1, \ldots, a_m, b\}$, the set of items is $\objSet = \{g_1, \ldots, g_m\} \cup \{h_i^j \mid i, j \in \llbracket 1, m \rrbracket, i \neq j \} \cup \{l_1, \ldots, l_{3m}\} \cup \{l_1^*, l_2^*\}$ and the preferences of the agent are given in the following table:
	\begin{center}
		\vspace{-0.5cm}
		\resizebox{\linewidth}{!}{
			$\begin{array}{c|ccc|ccccccc|ccc|cc}
			& g_1 & \ldots & g_m & h_1^2 & \ldots & h_1^m & \ldots & h_m^1 & \ldots & h_m^{m - 1} & l_1 & \ldots & l_{3m} & l_1^* & l_2^* \\
			\hline
			a_1 & \boxed{m + 1 - \epsilon} & 0 & 0 & \frac{m}{m - 1} & \frac{m}{m - 1} & \frac{m}{m - 1} & 0 & \boxed{0} & 0 & 0 & x_1 & \cdots & x_{3m} & \boxed{0} & 0 \\
			\vdots & 0 & \ddots & 0 & 0 & 0 & 0 & \ddots & 0 & 0 & 0 & x_1 & \cdots & x_{3m} & 0 & 0 \\
			a_m & 0 & 0 & \boxed{m + 1 - \epsilon} & 0 & 0 & \boxed{0} & 0 & \frac{m}{m - 1} & \frac{m}{m - 1} & \frac{m}{m - 1} & x_1 & \cdots & x_{3m} & 0 & \boxed{0} \\
			b & 0 & 0 & 0 & 0 & 0 & 0 & 0 & 0 & 0 & 0 & \boxed{x_1} & \boxed{\cdots} & \boxed{x_{3m}} & m & m
			\end{array}$}
	\end{center}
	where $\epsilon > 0$ is a constant small enough.
	
	Let us explain these preferences in detail. Agent $a_i$, $i \in \llbracket 1, m \rrbracket$, gives value $m + 1 - \epsilon$ to good $g_i$ and 0 to the other ``g'' goods. She values $\frac{m}{m - 1}$ every $h_i^j$ good and 0 the other ``h'' goods. Her valuation for ``l'' goods follows the value of the number in $X$ and finally $l_1^*$ and $l_2^*$ provide her no value. Agent $b$ only consider goods $l_1, \ldots, l_m, l_1^*$ and $l_2^*$ that are respectively valued as in $X$, $m$ and $m$.
	
	The initial allocation $\pi \in \Pi(\objSet, \agentSet)$, given as a entry of the \textsc{is-\GFoGood{}} problem, is presented by the boxed items in the previous table. It is defined as follow:
	\begin{align*}
	\pi_{a_1} & = \{g_1\} \cup \{h_i^1 \mid j \in \llbracket 2, m \rrbracket\} \cup \{l^*_1\}, \\
	\pi_{a_m} & = \{g_m\} \cup \{h_j^m \mid j \in \llbracket 1, m - 1 \rrbracket\} \cup \{l^*_2\}, \\
	\pi_{a_i} & = \{g_i\} \cup \{h_j^i \mid j \in \llbracket 1, m \rrbracket, j \neq i\}, \forall i \in \llbracket 2, m - 1 \rrbracket, \\
	\pi_{b} & = \{l_1, \ldots, l_{3m}\}.
	\end{align*}
	
	Next, we prove that allocation $\pi$ violates GEF1 if and only if there exists a partition $(X_i)_{i \in \llbracket 1, m \rrbracket}$ of $X$ satisfying the \textsc{3-Partition} conditions.
	
	We first show that if the allocation $\pi$ is not GEF1, then the two groups $S$ and $T$ witnessing this violation can only be $\agentSet$, that is $S = T = \agentSet$. To do so, let us consider any two groups $S$ and $T$ and a reallocation $\pi' \in \Pi(\pi_T, S)$ such that:
	$$\forall a \in S, \forall o \in \pi'_a, u_a(\pi_a'\backslash\{o\}) \geq u_a(\pi_a),$$
	with at least one inequality being strict. Such $S, T$ and $\pi'$ define a violation of GEF1.
	
	First note that $S = \{b\}$ is not a suitable witness as $b$ can not improve her utility by receiving one agent's entire bundle (remember that $|S| = |T|$).
	
	Let us then assume that there exist an agent $a_i$ in $S$, $i \in \llbracket 1, m \rrbracket$. Observe that if $\pi'_{a_i} = \objSet \backslash \{l_j \mid j \in \llbracket 1, 3m \rrbracket\}$, $u_{a_i} (\pi'_{a_i} \backslash \{g_i\}) = m$ which is strictly lower than her utility in $\pi$. Hence, to be better off in $\pi'$, $a_i$ should receive a subsets of the ``l'' items. This yields that we should have $b \in T$. Note also that if $\pi'_{a_i} = \{g_i\} \cup \{l_j \mid j \in \llbracket 1, 3m \rrbracket\}$, $u_{a_i} (\pi'_{a_i} \backslash \{g_i\}) = m$. Hence there should exist another agent $a_j$ in $T$, $j \in \llbracket 1, m \rrbracket, j \neq i$.
	
	For now, we have proved that there should be two agents in $T$, $a_j$ and $b$, and one in $S$, $a_i$. As $|T| = |S|$, another agent should be in $S$. 
	\begin{itemize}
		\item If an agent $a_{i'}$ is in $S$, $i'\in \llbracket 1, m \rrbracket, i' \neq i$, then $a_i$ and $a_{i'}$ have to share the ``l'' items than give a utility of $m$ in total. Hence, to be better off, they should receive some $h_i^k$ and $h_{i'}^{k'}$ items held by other ``a'' agents who should then be in $T$. This argument can then be iterated until having all ``a'' are in $T$ and in $S$. 
		\item If agent $b$ is in $S$, then to accept to give her ``l'' items, she should receive $l_1^*$ and $l_2^*$. Agents $a_1$ and $a_m$ should thus be in $T$. The previous argument can then be applied. This yields to the fact that the only possible groups witnessing a violation of GEF1 are $S = T = \agentSet$.
	\end{itemize}
	
	Let us now show that there exists a reallocation $\pi' \in \Pi(\objSet, S)$ that lead to a violation of GEF1, if and only if there exists a partition $(X_i)_{i \in m}$ of $X$ that respects the \textsc{3-Partition} conditions. Observe that the set of items $O_i = \{g_i\} \cup \{h_i^j \mid j \in \llbracket 1, m \rrbracket\}$ is given positive value only by agent $a_i$, hence $O_i \subseteq \pi'_i$. $O_i \backslash \{g_i\}$ brings $m - 1$ utility to $a_i$, hence to be better off, she needs to receive at least 1 additional unit of utility from the ``l'' items. Overall the ``l'' items can be divided so that it brings an additional 1 unit of utility to every agent $a_i$ if and only if there exist a suitable partition of $X$. If such partition $(X_i)_{i \in \llbracket 1, m \rrbracket}$ exists, then the reallocation $\pi'$ defined by:
	\begin{align*}
	\pi_{a_i}' & = \{g_i\} \cup \{h_i^j \mid j \in \llbracket 1, m \rrbracket, j \neq i\} \cup X_i, \forall i \in \llbracket 1, m \rrbracket, \\
	\pi_{b}' & = \{l_1^*, l_2^*\},
	\end{align*}
	is a witness of the violation of GEF1. Otherwise, $\pi$ satisfies GEF1.
	
	\medskip
	
	Finally this reduction is clearly done in polynomial-time which concludes the proof.
	
	\subsection{Proof of Theorem \ref{thm:complexityTestingGEF1} for \textsc{is-\GFoChore{}}}
	
	Let $X = \{x_1, \ldots, x_{3m}\}$ be an instance of the \textsc{3-Partition} problem. We present in the following its corresponding instance $(I, \pi)$ of the \textsc{is-\GFoChore{}} problem. The set of chores is $\objSet = \{g_1, \ldots, g_m\} \cup \{h_1, \ldots, h_m\} \cup \{l_1, \ldots, l_{3m}\} \cup \{o_1, \ldots, o_{2m}\}$ and the set of agents $\agentSet = \{a_1, \ldots, a_m\} \cup \{b_1, \ldots, b_m\}$. The utilities of the singletons are as follows.
	\begin{center}
		\vspace{-0.4cm}
		\resizebox{\linewidth}{!}{
			$\begin{array}{c|ccc|ccc|ccccccc|cccccc}
			& g_1 & \ldots & g_m & h_1 & \ldots & h_m & l_1 & l_2 & l_3 & \ldots & l_{3m - 2} & l_{3m - 1} & l_{3m} & o_1 & \ldots & o_m & o_{m + 1} & \ldots & o_{2m} \\
			\hline
			a_1 & \boxed{-m - \epsilon} & -M & -M & \boxed{-1 - \epsilon} & -M & 0 & -x_1 & -x_2 & -x_3 & \cdots & -x_{3m - 2} & -x_{3m - 1} & -x_{3m} & \boxed{0} & -M & -M & -M & -M & -M  \\
			\vdots & -M & \ddots & -M & -M & \ddots & -M & -x_1 & -x_2 & -x_3 & \cdots & -x_{3m - 2} & -x_{3m - 1} & -x_{3m} & -M & \ddots & -M & -M & -M & -M \\
			a_m & -M & -M & \boxed{-m - \epsilon} & -M & -M & \boxed{-1 - \epsilon} & -x_1 & -x_2 & -x_3 & \cdots & -x_{3m - 2} & -x_{3m - 1} & -x_{3m} & -M & -M & \boxed{0} & -M & -M & -M  \\
			\hline
			b_1 & -x_2 - x_3 & -M & -M & -M & -M & -M & \boxed{-x_1} & \boxed{-x_2} & \boxed{-x_3} & -M & -M & -M & -M & -M & -M & -M & \boxed{0} & -M & -M \\
			\vdots & -M & \ddots & -M & -M & -M & -M & -M & -M  & -M & \ddots & -M & -M & -M & -M & -M & -M & -M & \ddots & -M \\
			b_m & -M & -M & -x_{3m - 1} - x_{3m} & -M & -M & -M & -M & -M & -M & -M & \boxed{-x_{3m - 2}} & \boxed{-x_{3m - 1}} & \boxed{-x_{3m}} & -M & -M & -M & -M & -M & \boxed{0}
			\end{array}$}
	\end{center}
	where $\epsilon > 0$ is a constant small enough, $M$ is a constant greater than $m + 1$ and the $x_i$ are assumed to be ordered in a decreasing order: $\forall i \in \llbracket 1, 3m \rrbracket, x_i \geq x_{i + 1}$.
	
	Let us explain these preferences in detail. Agent $a_i$, $i \in \llbracket 1, m \rrbracket$, gives value $-m - \epsilon$ to chore $g_i$ and $-M$ ot any other ``g'' chore. Moreover, she values $-1-\epsilon$ chore $h_i$, 0 chore $h_{i - 1}$ (chore $h_m$ for agent $a_1$) and $-M$ the other ``h'' chores. Her valuation for ``l'' chores follows the opposite of the values in $X$. Finally, her utility for ``o'' chores is 0 for $o_i$ and $-M$ for the others.
	
	Agent $b_i, i \in \llbracket 1, m \rrbracket$, values $-x_{3i - 1} - x_{3_i}$ chores $g_i$, and $-M$ all other ``g'' chores. All ``h'' chores give her $-M$ utility. Chores $l_{3i - 2}, l_{3i - 1}$ and $l_{3i}$ respectively provide her $x_{3i - 2}, x_{3i - 1}$ and $x_{3i}$ utility while she values $-M$ any other ``l'' chore.  Finally her utility for ``o'' chores is 0 for $o_{i + m}$ and $-M$ for the others.
	
	The initial allocation $\pi \in \Pi(\objSet, \agentSet)$, given as a entry of the \textsc{is-\GFoChore{}} problem and represented by the boxed items in the previous table, is defined as follow:
	\begin{align*}
	\pi_{a_i} & = \{g_i\} \cup \{h_i\} \cup \{o_i\}, \forall i \in \llbracket 1, m \rrbracket, \\
	\pi_{b_i} & = \{l_{3i - 2}, l_{3i - 1}, l_{3i}\} \cup \{o_{m + i}\}, \forall i \in \llbracket 1, m \rrbracket.
	\end{align*}
	
	Next, we prove that allocation $\pi$ violates GEF1 if and only if there exists a partition $(X_i)_{i \in \llbracket 1, m \rrbracket}$ of $X$ satisfying the \textsc{3-Partition} conditions.
	
	We first show that if the allocation $\pi$ is not GEF1, then the two groups $S$ and $T$ witnessing this violation can only be $\agentSet$, that is $S = T = \agentSet$. To do so, let us consider any two groups $S$ and $T$ and a reallocation $\pi' \in \Pi(\pi_T, S)$ such that:
	\begin{align}
	\label{eq:complexChoreLine1}
	\forall a \in S, \forall o \in \pi_a, u_a(\pi'_a) \geq u_a(\pi_a \backslash \{o\}),
	\end{align}
	with at least one inequality being strict. Such $S, T$ and $\pi'$ define a violation of GEF1. Note that as \eqref{eq:complexChoreLine1} should hold for any item $o \in \pi_a$, it should hold in particular if $o$ is the worst chore in $\pi_a$ that is $g_i$ for agents $a_i$ and $l_{3i - 2}$ for agent $b_i$ (remember that we assumed $x_i$ are decreasingly ordered). Hence \eqref{eq:complexChoreLine1} can be reformulated as:
	\begin{align}
	\label{eq:complexChoreLine2}
	\forall a \in S, \left\{\begin{array}{ll}
	u_a(\pi'_a) \geq -1 - \epsilon & \text{if } a \in \{a_1, \ldots, a_m\}, \\
	u_a(\pi'_a) \geq -x_{3i - 1} - x_{3i} & \text{if } a = b_i, i \in \llbracket 1, m \rrbracket,
	\end{array}\right.
	\end{align}
	with one inequality being strict. 
	
	\medskip
	
	Let us show that \eqref{eq:complexChoreLine2} can be satisfied if and only if $S = T = \agentSet$. To do so, we first prove that $S = T$. Then, we show that the number of ``a'' agents in $T$ is equal to the number of ``b'' agents in $T$. Finally we prve that if agent $a_i$ (resp. $a_1$) is in $T$, then agent $a_{i - 1}$ (resp. $a_m$) should also be in $T$, proving that $\{a_1, \ldots, a_m\} \in T$. All these facts put together show the claim.
	
	Consider agent $c \in T$ that can either be an ``a'' or ``b'' agent. Let $k$ be the index such that $k = i$ if $c = a_i$ and $k = i + m$ if $c = b_i$. Observe that chore $o_k$ is in $\pi_c$ and it can only be allocated to agent $c$ for \eqref{eq:complexChoreLine2} to be satisfied as every other agent have utility $-M$ for it. We have thus $T \subseteq S$, and as $|S| = |T|$ we have $S = T$.
	
	Let us then introduce two notations, $n_a$ and $n_b$, respectively corresponding to the number of ``a'' agents and ``b'' agent in $T$, that is $n_a = |\{a_1, \ldots, a_m\} \cap T|$ and $n_b = |T| - n_a$. Note that if $a_i \in T$ then $b_i$ should be in $S$ as $g_i \in \pi_{a_i}$ can only be allocated to $b_i$ to satisfy \eqref{eq:complexChoreLine2}. Remember that $S = T$, hence $b_i \in T$ and thus $n_a \leq n_b$.
	
	Moreover, observe that $T \subseteq \{b_1, \ldots, b_m\}$ would violated \eqref{eq:complexChoreLine2} as ``b'' agents are interested only in the ``l'' chores they own among all ``l'' chores. Hence, no reallocation among only ``b'' agents can be improving. Let us then assume that there exists agent $a_i$ in $T$. As already mentioned, we have $g_i \in \pi'_{b_i}$, hence $\pi'_{b_i} = \{g_i, o_{i + m}\}$ and \eqref{eq:complexChoreLine2} for $b_i$ is at equality. All ``l'' chores in $\pi_T$ should then be reallocated to ``a'' agents as ``b'' can not received any additional chore. Remember that $\forall x \in X, 1/4 < x < 1/2$, hence any subset of $X$ of size strictly greater than 4 has a sum strictly greater than 1. For a suitable $\epsilon$ it is then impossible to reallocate more than $3$ ``l'' items to an ``a'' agents without its utility being below $-1 - \epsilon$. This implies that $n_b \leq n_a$ hence $n_a = n_b$. 
	
	In addition, there are exactly $3n_b$ ``l'' chores in $\pi_T$ and as $n_a = n_b$, each ``a'' agent should be reallocated exactly 3 ``l'' items. It is then not possible for agent $a_i$ to receive chore $h_i$, this chore should then be allocated to agent $a_{i - 1}$ ($a_m$ for $a_i = a_1$). Hence, if $a_i \in T$ we should have $a_{i - 1} \in S$ ($a_m$ for $a_i = a_1$) and thus in $T$ as $S = T$. This yields that all ``a'' agents are in $T$ and thus also all ``b'' agents as $n_a = n_b$. We have thus $S = T = \agentSet$.
	
	\medskip
	
	We now prove that for $S = T = \agentSet$, there exist a reallocation $\pi'$ satisfying \eqref{eq:complexChoreLine2} if and only if there exists a partition $(X_i)_{i \in \llbracket 1, m \rrbracket}$ of $X$ satisfying the conditions of the \textsc{3-Partition} problem.
	
	Note that each $g_i$ chore should be allocated to $b_i$ in $\pi'$, hence all ``l'' chores should be divided among ``a'' agents. ``h'' chores are allocated to ``a'' agents receiving 0 utility for it, similar reallocation is done for ``o'' chores. Hence \eqref{eq:complexChoreLine2} is satisfied if and only if it is possible to divide the ``l'' items into $m$ parts of sum smaller than $1 + \epsilon$, that is of sum 1. For a suitable $\epsilon$ this is equivalent to the existence of a partition of $X$ satisfying the conditions of the \textsc{3-Partition} problem. 
	
	\medskip
	
	Finally this reduction is clearly done in polynomial-time which concludes the proof.
	
	\subsection{Proof of Theorem \ref{thm:complexityTestingGEF1} for \textsc{is-GEF1}}
	
	Since any instance of the \textsc{is-\GFoGood{}} or \textsc{is-\GFoChore{}} problems is also an instance of the \textsc{is-GEF1} problem, the previous also prove that \textsc{is-GEF1} is coNP-complete.
	
\end{document}